\documentclass[journal]{IEEEtran}
\usepackage{fancyhdr}
\usepackage{amsmath,amssymb}
\usepackage{amsfonts}
\usepackage{dsfont}
\usepackage{balance}
\usepackage{algpseudocode}
\usepackage{algorithm}
\usepackage{tikz}
\usepackage[hidelinks]{hyperref}    %for using hyperlinks
\usepackage{color}
\usepackage{pgfplots}           %for plotting graphs
\pgfplotsset{compat=1.16}
\usepackage{graphicx}
\usepackage{cite}
\usepackage{mathrsfs}
\usepackage{caption}
\usepackage{subcaption}
\usepackage{stfloats}
\usepackage{tikz}
\usepackage{mathdots}
\usepackage{yhmath}
\usepackage{cancel}
\usepackage{color}
\usepackage{siunitx}
\usepackage{array}
\usepackage{xcolor}
\usepackage{multirow}
\usepackage{textcomp}
\usepackage{gensymb}
\usepackage[normalem]{ulem}
\usepackage{booktabs}
\usetikzlibrary{fadings}
\usetikzlibrary{patterns}
\usetikzlibrary{shadows.blur}
\usepackage{braket}
\usepackage{enumitem}
\usepackage{tabularx}
\usepackage{array}   % for \newcolumntype
\newcolumntype{Y}{>{\raggedright\arraybackslash}X}

\newtheorem{lemma}{Lemma}

\newtheorem{proposition}{\textbf{Proposition}}

\usepackage{amsmath}
\usetikzlibrary{arrows.meta,calc}
\DeclareMathOperator\erf{erf}
\DeclareMathOperator{\erfc}{erfc}

\begin{document}
%\IEEEoverridecommandlockouts

\title{Error Probability Analysis of Quantum Communication with Phase-squeezed M-PSK}
\author{Nikos A. Mitsiou, Member, IEEE and Ioannis Krikidis, Fellow, IEEE 

\thanks{Nikos A. Mitsiou and Ioannis Krikidis are with the Department of Electrical and Computer Engineering, University of Cyprus, 1678 Nicosia, Cyprus (e-mail: nmitsi02@ucy.ac.cy, krikidis@ucy.ac.cy).}

\vspace{-0.4cm}}

\maketitle
\begin{abstract}
In this paper, we investigate the symbol error probability (SEP) of phase-squeezed \(M\)-ary phase-shift keying (\(M\)-PSK). Since the relevant observable for \(M\)-PSK detection is the optical phase, we adopt the adaptive Mark-II receiver which is a physically realizable phase measurement. First, we develop a theoretical analysis based on the phase probability operator measure (POM) of the Mark-II scheme in the Fock basis. Then, we develop two SEP methods based on the statistics of the received PSK symbol and the error introduced by the Mark-II measurement. The first method derives the phase probability density induced by the squeezed state noise and incorporates the additional Mark-II phase uncertainty through an angular convolution. Since this convolution does not admit a simple closed form, we also introduce an effective tangential-variance model, which yields a closed form SEP expression in terms of the Owen's \(T\)-function. Numerical results show that phase squeezing substantially reduces the SEP of \(M\)-PSK compared to coherent state transmission, with greater gains for higher constellation orders. Notably, for the investigated scenario, squeezing can almost double the photon efficiency of \(M\)-PSK as the mean number of transmitted photons increases. Finally, the proposed approximations closely follow the Mark-II POM analysis, typically within an accuracy of 2-4 photons, and therefore provide accurate and computationally efficient tools for analyzing phase-squeezed quantum \(M\)-PSK communication.
% This formulation expresses the SEP through a weighted Fourier series, whose weights are determined by the transmitted quantum state coefficients and the Hermitian Mark-II measurement matrix.
\end{abstract}
\section{Introduction}
Quantum optical communication studies the transmission of information using  photons to encode classical information into quantum states of light which can be transmitted over optical fibers and free-space optical (FSO) channels \cite{Shapiro2009QuantumTheoryOpticalCommunications}. In particular, quantum FSO communication combines the advantages of quantum optical signal processing with the flexibility of the wireless medium and is relevant to application scenarios such as ground-to-satellite and satellite to ground optical links, as well as terrestrial FSO systems operating under strong atmospheric turbulence or stringent power constraints. More broadly, quantum FSO communication constitutes an important physical layer building block for emerging technologies such as quantum enhanced sensing and quantum key distribution, low-power optical Internet of Things, and hybrid classical quantum communication networks
\cite{Bedington2017ProgressSatelliteQKD, hanzo}. In these scenarios, path loss, limited received optical power, and atmospheric turbulence may drive the system into the photon-limited regime, where performance is fundamentally constrained by quantum noise \cite{Vasylyev2016AtmosphericQuantumChannels}. 

Coherent states of light, are recognized as the quantum mechanical states that most closely resemble the classical waves used in classical optical communications. They are minimum uncertainty Gaussian states whose mean field amplitude can be displaced to represent information symbols \cite{Kato2022}. Coherent states are widely deployed, but their isotropic noise distribution in the phase space can produce significant overlap between different states, which can restrict their distinguishability and the achievable information rates under a limited photon budget \cite{Kato2025}. As such, nonclassical Gaussian states, and in particular squeezed states, provide a promising alternative optical modulation. Squeezed states are particularly attractive because they reduce the uncertainty in one field quadrature below the coherent state level at the expense of increased uncertainty in the conjugate quadrature \cite{Weedbrook2012GaussianQuantumInformation}. This redistribution of noise makes squeezed states natural candidates for scenarios in which detection performance is quantum noise limited. 

Apart from modulation, quantum measurements are also an essential part of quantum communication because they convert the transmitted quantum state into a classical observation used for symbol detection. In optical receivers, homodyne detection measures a single field quadrature by interfering the received mode with a strong local oscillator, while heterodyne detection jointly measures both conjugate quadratures at the cost of additional measurement noise \cite{Weedbrook2012GaussianQuantumInformation}. These quadrature measurements are naturally suitable for amplitude and quadrature modulated signals. By contrast, for phase-modulated constellations such as \(M\)-PSK, the degree of freedom that carries information is the optical phase. Thus, the measurement that naturally aligns with $M$-PSK is a phase measurement. Since an ideal phase measurement is generally only a theoretical benchmark rather than a direct physical receiver, the family of adaptive dyne measurements was proposed, that provide an implementable approach to optical phase estimation \cite{Wiseman1998AdaptiveSingleShotPhaseMeasurements}. For instance, the adaptive Mark-II measurement uses feedback controlled local oscillator phases to approximate an ideal phase measurement, while introducing a phase error that asymptotically decays in relation to the mean number of photons of the quantum state.

\subsection{Literature Review}
Coherent state modulation constitutes a baseline framework for quantum
optical communication, since coherent states are experimentally accessible and well studied in the literature.
Early quantum detection theory established the fundamental limits of
discriminating nonorthogonal optical states
\cite{YuenKennedyLax1975,Helstrom1976QuantumDetection}, while receiver
designs such as the Kennedy and Dolinar receivers showed how structured
optical measurements can approach or attain optimal discrimination for coherent states \cite{kennedy1973near,Dolinar1973Receiver}.
PSK and quadrature amplitude modulation (QAM) constellations based on coherent states have
been studied both in terms of quantum detection, mutual information, and
non-orthogonality of the transmitted signal set
\cite{kato,Kato2022,Kato2025}. This line of work has also led to
the development of novel experimental and theoretical coherent receiver architectures for multi-state discrimination, including receivers that beat the standard
quantum limit for non-orthogonal PSK alphabets
\cite{Shapiro2009QuantumTheoryOpticalCommunications,TakeokaSasaki2008,
Tsujino2011QuantumReceiver,Becerra2013MPSK,Becerra2015MPSK}. These
studies established coherent state modulation as the reference point for studying quantum communications.

Beyond coherent states, squeezed and other nonclassical Gaussian states
provide additional degrees of freedom for quantum communication by shaping
the quantum uncertainty of the optical field. In the broader continuous-variable framework, Gaussian quantum states opened the way to a wide variety of tasks and applications, including quantum communication, quantum cryptography, quantum computation, quantum teleportation, and quantum state and channel discrimination \cite{BraunsteinVanLoock2005,Weedbrook2012GaussianQuantumInformation}. In particular, squeezed states were
introduced as two-photon coherent states and later recognized as a
fundamental resource for reducing quantum noise in optical measurements
\cite{Yuen1976TwoPhoton,Caves1981QuantumMechanicalNoise,
Slusher1985SqueezedLight}.  For communication systems, squeezing can redistribute noise between
conjugate quadratures, making it potentially useful when the decision
geometry is more sensitive to one field direction than the other. This
principle has been explored in quantum modulation and quantum state discrimination, including squeezed binary PSK (BPSK) communication
\cite{Chesi2018SqueezingEnhancedPSK}, non-Gaussian photon-added Gaussian
modulation \cite{10437870}, optimized squeezing for PSK state
discrimination \cite{Bhadani2022OptimizedSqueezingPSKQSD}, and
displaced-squeezed transmission over turbulence channels
\cite{Krikidis2026QuantumRotationDiversity}. Specifically, \cite{Bhadani2022OptimizedSqueezingPSKQSD} showed that for displacement receivers with optimized parameters, phase squeezing at the transmitter does not improve the error probability of BPSK and quadrature PSK (QPSK). Finally, recent experimental work has also compared coherent and squeezed light sources for direct quantum
communication, showing that squeezed states can improve communication
reliability and security \cite{Paparelle2025ExperimentalDirectQSDC}.

\section{Motivation \& Contribution}
Despite the existing literature, the effect of phase squeezing in \(M\)-PSK
constellations with \(M>4\) remains unexplored. The conclusion
reported in \cite{Bhadani2022OptimizedSqueezingPSKQSD} does not rule out
a benefit of phase squeezing in general, since that study was restricted to
BPSK and QPSK only. For low-order $M$-PSK, the correct detection decision boundaries
are orthogonal. Therefore, for BPSK and QPSK, any gain obtained by reducing the uncertainty along the angular
direction is offset by the increased uncertainty along the radial direction. However, for higher-order \(M\)-PSK, the geometry is different. The main error
mechanism is phase boundary crossing between neighboring symbols, and coherent noise distribution, i.e. circular distribution, is not necessarily optimal. This creates the possibility of assigning a symbol dependent squeezing angle, so that the low variance quadrature is
aligned with the local phase sensitive direction of each PSK symbol, while
the increased uncertainty is pushed mainly into the radial direction. Therefore, this
geometric intuition, shown in  Fig. \ref{fig:psk}, motivates a dedicated analysis for quantum communication with phase squeezing for 
higher-order \(M\)-PSK constellations.

As a consequence, this paper investigates the SEP of
phase squeezed \(M\)-PSK under a physically realizable phase receiver.
Since the available degree of freedom in \(M\)-PSK is the
optical phase, we adopt the adaptive Mark-II phase measurement and utilize
its phase probability operator measure (POM) as the quantum receiver
model. The proposed constellation is based on displaced squeezed states
with symbol dependent squeezing angles, chosen so that the low variance
quadrature is aligned with the local phase sensitive direction of each
PSK symbol, while the increased uncertainty is pushed into the
radial direction. The objective is to determine whether this
symbol dependent phase squeezing can improve the distinguishability of
higher order PSK constellations under a fixed photon budget, and to
develop both exact and tractable analytical tools for quantifying the
resulting SEP. The main contributions of the paper are summarized as
follows

\begin{itemize}
  \item We derive a SEP expression for the proposed
  phase-squeezed \(M\)-PSK constellation under the adaptive Mark-II phase
  measurement. The analysis uses the Mark-II phase POM
  and expresses the measured phase distribution of the received state as a weighted Fourier
  series determined by the Mark-II matrix and the received state's coefficients in the Fock basis.

  \item  Using a statistical analysis in the polar coordinates, a closed-form expression is derived for the phase distribution of the received quantum state, while the total measured phase distribution and the corresponding SEP are obtained numerically through angular convolution
    with the Mark-II phase error distribution.

  \item We also derive a closed-form SEP approximation by incorporating the
  Mark-II phase uncertainty into an effective tangential variance. This
  leads to a SEP in terms of Owen's
  \(T\)-function and gives direct insights into the joint effects of
  displacement, radial anti-squeezing, tangential noise reduction, and
  receiver phase uncertainty.

  \item Numerical results compare the Mark-II POM benchmark, the
  polar-domain analysis, the Owen's-\(T\) approximation, and the coherent state
  \(M\)-PSK. The results show that phase squeezing substantially reduces
  the SEP for higher-order constellations, while QPSK has zero optimal
  squeezing, in agreement with prior conclusions for low-order PSK
  \cite{Bhadani2022OptimizedSqueezingPSKQSD}. Both approximations closely
  follow the POM analysis, typically within an accuracy of
  \(2\)-\(4\) photons, with the polar-domain method being slightly more
  accurate due to its explicit angular-convolution modeling.
\end{itemize}

\section*{Structure}
The rest of the paper is organized as follows. Section~III introduces the
coherent and phase-squeezed \(M\)-PSK constellation under the fixed
total-photon constraint. Section~IV describes the phase-measurement model,
including the canonical phase POM, the adaptive Mark-II POM, and the
corresponding Fock-basis SEP expression. Section~V develops the
polar-domain SEP analysis, where the state-induced phase distribution is
derived and the Mark-II phase error is incorporated through angular
convolution. Section~VI presents the effective tangential-variance model
and derives the closed-form SEP approximation involving Owen's
\(T\)-function. Section~VII provides numerical results comparing the full
POM benchmark, the proposed approximations, and coherent state
transmission. Finally, Section~VIII concludes the paper and discusses
future research directions.

\section*{Notation}
Throughout the paper, \(j=\sqrt{-1}\) denotes the imaginary unit. For a complex number, \((\cdot)^*\), \(\Re\{\cdot\}\), \(\Im\{\cdot\}\), and \(|\cdot|\) denote complex conjugation, real part, imaginary part, and modulus, respectively. For an operator, \((\cdot)^\dagger\) denotes the Hermitian adjoint, and \([\cdot,\cdot]\) denotes the commutator. We use the standard Dirac notation; \(|\cdot\rangle\) denotes a ket, \(\langle\cdot|\) denotes a bra, \(\langle\cdot|\cdot\rangle\) denotes an inner product, and \(|\cdot\rangle\langle\cdot|\) denotes an outer product. The trace operator is denoted by \(\operatorname{Tr}\{\cdot\}\), and expectation is denoted by \(\langle\cdot\rangle\) or \(\mathbb{E}\{\cdot\}\). The probability of an event is denoted by \(\Pr\{\cdot\}\), while \(f_X(\cdot)\) denotes the probability density function of a random variable \(X\). The Gaussian \(Q\)-function is denoted by \(Q(\cdot)\), while \(\Phi(\cdot)\) and \(\phi(\cdot)\) denote the standard normal cumulative distribution function and probability density function, respectively. The error function, the complementary error function, and Owen's \(T\)-function are denoted by \(\operatorname{erf}(\cdot)\), \(\operatorname{erfc}(\cdot)\), and \(T(\cdot,\cdot)\), respectively. The determinant, floor, double-factorial, and correlation operators are written as \(\det(\cdot)\), \(\lfloor\cdot\rfloor\), \((\cdot)!!\), and \(\operatorname{Corr}(\cdot,\cdot)\), respectively. The indicator function is denoted by \(\mathbf{1}\{\cdot\}\). All angular quantities are understood modulo \(2\pi\) unless otherwise stated.

\section{System Model}
This section introduces the quantum optical constellations used throughout the paper, starting from coherent state \(M\)-PSK, which is the baseline constellation, and then extending it to the proposed phase-squeezed \(M\)-PSK.

\subsection{Coherent $M$-PSK}
We consider a single mode with annihilation and creation
operators $\hat a$ and $\hat a^\dagger$ satisfying the commutation property
$[\hat a,\hat a^\dagger]=1$.
The photon-number operator is
$ \hat n = \hat a^\dagger \hat a$. A coherent state is then obtained from the vacuum through the displacement
operator
\begin{equation}
  |\alpha\rangle = \hat D(\alpha)\,|0\rangle,
  \qquad
  \hat D(\alpha)=\exp\!\bigl(\alpha \hat a^\dagger-\alpha^* \hat a\bigr),
\end{equation}
and satisfies $\hat a |\alpha\rangle = \alpha |\alpha\rangle,
  \,
  \langle \alpha|\hat n|\alpha\rangle = |\alpha|^2$
\cite{Krikidis2026QuantumRotationDiversity}. An $M$-ary coherent state PSK constellation is formed by fixing the
modulation photon number $N_s>0$ and defining
\begin{equation}
  \alpha_m=\sqrt{N_s}\,e^{j\phi_m},
  \qquad
  \phi_m=\frac{2\pi m}{M},
  \quad m=0,\dots,M-1.
\end{equation}
Equivalently, if $|\alpha_0\rangle = |\sqrt{N_s}\rangle$ is a reference
state, each PSK symbol can be written as
\begin{equation}
  |\psi_m\rangle
  = \hat R(\phi_m)\,|\alpha_0\rangle
  = |\alpha_m\rangle
  = \hat D(\alpha_m)\,|0\rangle,
\end{equation}
where \(\hat R(\phi)=e^{j\phi\hat n}\) is the phase-space rotation operator, with
\(\hat R(\phi)|\alpha\rangle=|\alpha e^{j\phi}\rangle\) \cite{kato}. Hence, the total PSK constellation is the following set of states
\begin{equation}
  \mathcal S_{\mathrm{PSK}}
  =
  \bigl\{
    |\psi_m\rangle
    =
    |\sqrt{N_s}\,e^{j\phi_m}\rangle
  \bigr\}_{m=0}^{M-1},
\end{equation}
in which every symbol has a  mean photon number equal to $ \langle \psi_m|\hat n|\psi_m\rangle = |\alpha_m|^2 = N_s$.
\subsection{Phase-squeezed $M$-PSK}
We now construct a phase-squeezed PSK constellation under a fixed total
photon budget. Single-mode squeezing is described by
\begin{equation}
  \hat S(\zeta)
  =
  \exp\!\left[
    \frac{1}{2}
    \bigl(
      \zeta^* \hat a^2 - \zeta \hat a^{\dagger 2}
    \bigr)
  \right],
\end{equation}
with squeezing parameter $ \zeta = r e^{j\theta_s},
  \, r\ge 0$, where $r$ denotes the squeezing amplitude and $\theta_s$ denotes the squeezing angle \cite{Krikidis2026QuantumRotationDiversity}. The squeezed vacuum
$|\mathrm{SV}(\zeta)\rangle = \hat S(\zeta)|0\rangle$
contains
\begin{equation}
  \langle \mathrm{SV}(\zeta)|\hat n|\mathrm{SV}(\zeta)\rangle
  =
  \sinh^2 r
\end{equation}
photons on average. We define the displaced-squeezed state as
\begin{equation}
  |\alpha,\zeta\rangle
  =
  \hat D(\alpha)\,\hat S(\zeta)\,|0\rangle .
\end{equation}
With this operator ordering, the total mean photon number per transmitted symbol is
\begin{equation} \label{eq:Ntot_split_psk}
  N_{\mathrm{tot}} = \langle \alpha,\zeta|\hat n|\alpha,\zeta\rangle
  =
  |\alpha|^2 + \sinh^2 r,
\end{equation}
where $|\alpha|^2$ are the displacement photons and
$\sinh^2 r$ the squeezing photons. We note that throughout the paper, whenever the optimal value of squeezing $r^*$ is required with respect to the energy constraint of \eqref{eq:Ntot_split_psk}, it is calculated based on a simple exhaustive search method. To proceed, we assign to each PSK symbol a symbol-dependent squeezing parameter
\begin{equation}
  \zeta_m = r e^{j\theta_{s,m}} .
\end{equation}
For single-mode squeezing, the ellipse orientation is set by half the
phase of $\zeta_m$. Therefore, alignment with the symbol direction
$\phi_m$ requires
\begin{equation}
  \theta_{s,m}=2\phi_m,
\end{equation}
up to an additional phase of $\pi$, which swaps the squeezed and
anti-squeezed axes.
In our implementation, we use
\begin{equation}
  \zeta_m = -\,r\,e^{j2\phi_m},
\end{equation}
so that the tangential direction, i.e. the phase-sensitive direction for
PSK, is squeezed, while the radial direction is anti-squeezed. The transmitted state for symbol $m$ is then
\begin{equation}
  |\psi_m\rangle
  =
  \hat D(\alpha_m)\,\hat S(\zeta_m)\,|0\rangle,
  \qquad m=0,\dots,M-1.
\end{equation}
Geometrically, the mean amplitudes $\alpha_m$ still lie on a circle of
radius $\sqrt{N_s}$, as in standard $M$-PSK, but each symbol is now
surrounded by a rotated noise ellipse. A PSK constellation with coherent and phase-squeezed states is shown in Fig.~\ref{fig:psk}.

\begin{figure}[t]
    \centering
    \begin{subfigure}{.8\columnwidth}
        \centering
        \includegraphics[width=\linewidth]{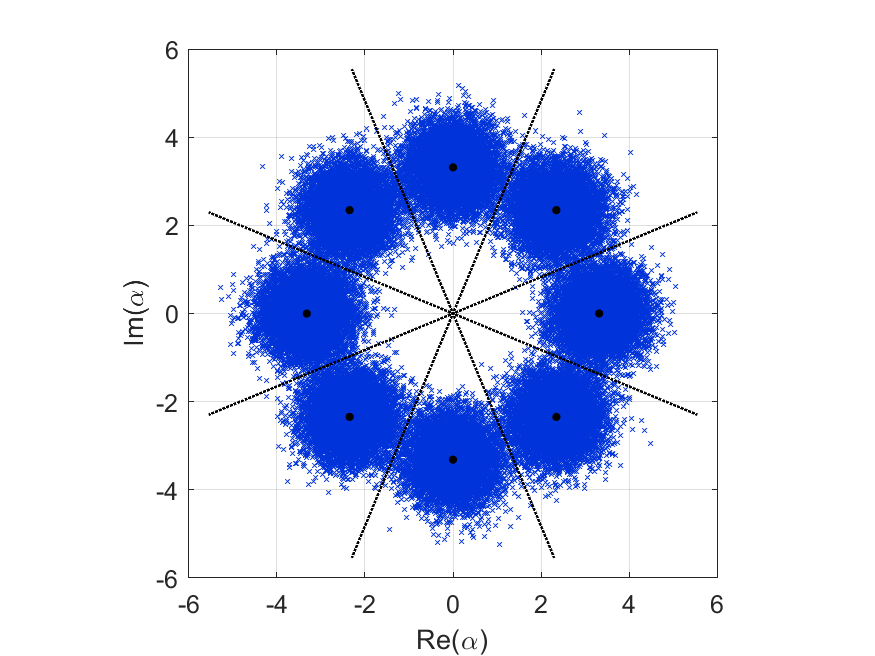}
        \caption{PSK with coherent transmission}
        \label{fig:pdf1}
    \end{subfigure}
    \begin{subfigure}{.8\columnwidth}
        \centering
        \includegraphics[width=\linewidth]{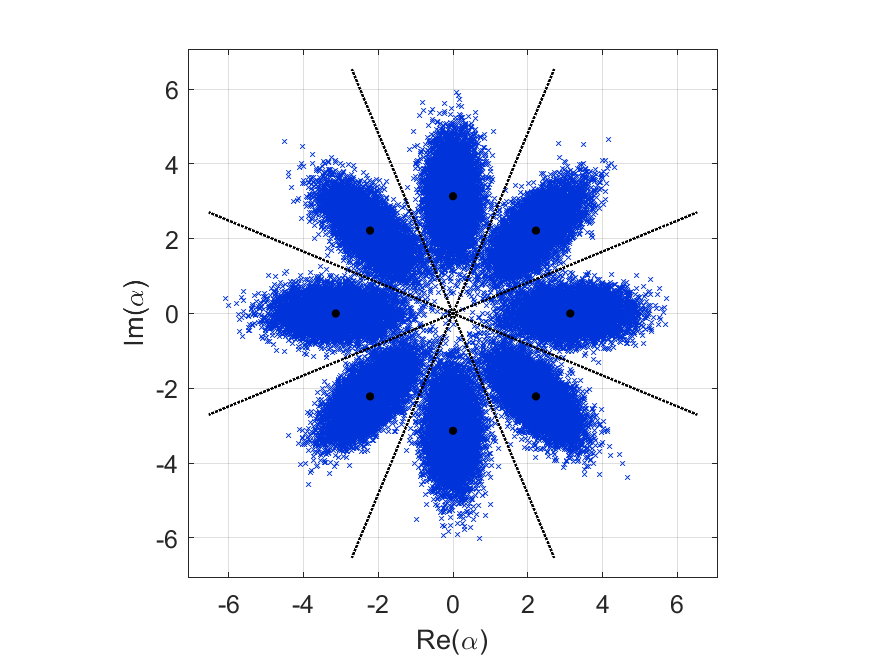}
        \caption{PSK with phase squeezing, $r=0.4$.}
        \label{fig:pdf2}
    \end{subfigure}
    \caption{PSK constellations, $M=8,\,N_\mathrm{tot}=10$.}
    \label{fig:psk}
    \vspace{-0.4cm}
\end{figure}

\section{Quantum Phase Measurement}
In this section, we describe the phase measurement used to evaluate the
performance of the proposed phase squeezed PSK modulation. The purpose of this measurement is to convert the received quantum state into a probability
density over all possible optical phases. We note that the notion of phase measurement naturally aligns with our scheme, because the transmitted information is encoded in the angle of the transmitted state in the phase space.

An important difference from the usual quantum measurements is that the phase of a quantum state
does not fit the usual projective measurement description in a
direct way. In a standard projective measurement, an observable is represented
by a Hermitian operator. The possible measurement outcomes are the eigenvalues of
that operator, while the associated eigenvectors, determine the projectors used to compute the probabilities of
those outcomes through the Born rule. For example, photon-number measurement is
described by the number operator with its eigenstates being the Fock states
\(\{|n\rangle\}_{n=0}^{\infty}\), and its possible outcomes being the integers
\(n\).

The phase of a quantum state is different. There is no standard Hermitian phase operator
with a complete set of normalizable, orthogonal phase eigenstates playing the
same role as the Fock states for photon number. Therefore, phase is most
naturally described by a phase POM
\cite{Wiseman1998AdaptiveSingleShotPhaseMeasurements}. In this description,
the measurement is specified by a family of positive operators \(F(\phi)\),
indexed by the continuous phase label \(\phi\in[0,2\pi)\). These operators
define the probability density for observing each phase outcome. This
viewpoint is illustrated in Fig.~\ref{fig:phase_measurement}. The input to the phase measurement is
the quantum state of the received symbol, and the output is a phase probability density function (PDF). The peak of this density indicates the most likely
measured phase, which ideally would be the phase of the transmitted quantum state, while its width represents the phase uncertainty caused by the
receiver and noise.

Formally, let us denote with \(\rho\) the density operator of the transmitted symbol. A phase
measurement is specified by a POM \(F(\phi)\) satisfying
\begin{equation}
  \int_0^{2\pi} F(\phi)\,d\phi = \mathbb{I},
  \label{eq:phase_pom_norm}
\end{equation}
where \(\mathbb{I}\) is the identity operator. This condition ensures that
the total probability over all phase outcomes is one. Then, the corresponding phase probability density is
\begin{equation}
  p(\phi)=\mathrm{Tr}\!\big[\rho\,F(\phi)\big].
  \label{eq:phase_PDF_general}
\end{equation}
Thus, for the PSK receiver, the object used for decision making is not a
single measured phase value alone, but the distribution \(p(\phi)\) induced by
the selected phase measurement.

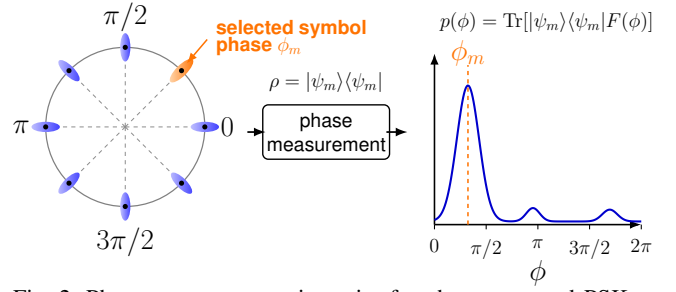
\begin{figure}[t]
    \centering
    \resizebox{\columnwidth}{!}{%
        \begin{tikzpicture}[
    x=1cm,y=1cm,
    >=Latex,
    font=\Large,
    title/.style={font=\sffamily\bfseries\Large},
    mathlabel/.style={font=\huge},
    smallmath/.style={font=\Large},
    orangeText/.style={orange!90!red, font=\sffamily\bfseries\Large},
    flowarrow/.style={->, line width=1.25pt, black},
    guide/.style={black!45, dashed, line width=0.65pt},
    ring/.style={black!50, line width=0.90pt},
    axis/.style={->, line width=1.15pt, black},
    pdfcurve/.style={blue!80!black, line width=1.50pt},
    selected/.style={orange!90!red}
]

% -------------------------------------------------
% Enlarged bounding box
% -------------------------------------------------
\path[use as bounding box] (0,0) rectangle (16.75,7.85);

% =================================================
% (a) Phase-squeezed PSK constellation
% =================================================
\begin{scope}[shift={(2.85,3.68)}]
    \def\R{2.02}

    \draw[ring] (0,0) circle (\R);

    \foreach \a in {0,45,90,135,180,225,270,315}{
        \draw[guide] (0,0) -- (\a:\R);
    }

    % Phase-squeezed clouds: enlarged, radially elongated, angularly narrow
    \foreach \a in {0,90,135,180,225,270,315}{
        \begin{scope}[shift={(\a:\R)}, rotate=\a]
            \shade[
                left color=blue!25,
                right color=blue!75,
                opacity=0.75
            ] (0,0) ellipse (0.36 and 0.13);
            \fill[black] (0,0) circle (0.065);
        \end{scope}
    }

    % Selected symbol: enlarged
    \begin{scope}[shift={(45:\R)}, rotate=45]
        \shade[
            left color=orange!30,
            right color=orange!90!red,
            opacity=0.83
        ] (0,0) ellipse (0.40 and 0.14);
        \fill[black] (0,0) circle (0.065);
    \end{scope}

    % Angle labels
    \node[mathlabel, anchor=west]  at (\R+0.28,0) {$0$};
    \node[mathlabel, anchor=east]  at (-\R-0.28,0) {$\pi$};
    \node[mathlabel, anchor=south] at (0,\R+0.28) {$\pi/2$};
    \node[mathlabel, anchor=north] at (0,-\R-0.40) {$3\pi/2$};

    % Selected annotation, shifted right/up to avoid constellation labels
    \draw[->, selected, line width=1.10pt]
        (2.08,2.22) -- ($(45:\R)+(0.13,0.18)$);

    \node[orangeText, anchor=west] at (2.18,2.48) {selected symbol};
    \node[orangeText, anchor=west] at (2.18,2.06) {phase $\phi_m$};
\end{scope}

% =================================================
% (b) Phase measurement
% =================================================
\node[smallmath] at (7.95,4.82)
    {$\rho = |\psi_m\rangle\langle\psi_m|$};

\node[
    draw=black,
    rounded corners=4pt,
    minimum width=2.78cm,
    minimum height=1.38cm,
    line width=1.00pt,
    font=\sffamily\Large,
    align=center,
    inner sep=3pt
] (meas) at (7.95,3.55) {phase\\measurement};

\draw[flowarrow] (5.95,3.55) -- (6.42,3.55);
\draw[flowarrow] (9.48,3.55) -- (10.02,3.55);

% =================================================
% (c) Output: phase probability density
% =================================================
\node[smallmath] at (13.55,6.35)
    {$p(\phi)=\mathrm{Tr}\!\left[|\psi_m\rangle\langle\psi_m|F(\phi)\right]$};

\begin{scope}[shift={(10.70,1.20)}]
    \def\W{5.25}
    \def\H{4.35}

    % Main PDF peak placed between 0 and pi/2, approximately near pi/4.
    \def\xm{0.85}

    % Axes
    \draw[axis] (0,0) -- (0,\H);
    \draw[axis] (0,0) -- (\W,0);

    % y-label: moved left enough to avoid sitting on the arrow
    % \node[mathlabel, anchor=east] at (-0.16,3.02) {$p(\phi)$};

    % x-label
    \node[mathlabel, anchor=north] at (2.62,-0.78) {$\phi$};

    % Tick marks and labels
    \foreach \x/\lab in {
        0/{0},
        1.3125/{\pi/2},
        2.625/{\pi},
        3.9375/{3\pi/2},
        5.25/{2\pi}
    }{
        \draw[black, line width=0.78pt] (\x,0) -- (\x,-0.12);
        \node[smallmath, anchor=north] at (\x,-0.25) {$\lab$};
    }

    % Smooth PDF curve
    \draw[pdfcurve, domain=0:5.25, samples=120, smooth]
        plot (
            \x,
            {
              0.08
              + 3.45*exp(-((\x-\xm)/0.40)^2)
              + 0.34*exp(-((\x-2.50)/0.24)^2)
              + 0.30*exp(-((\x-4.45)/0.27)^2)
            }
        );

    % Selected phase marker
    \draw[selected, dashed, line width=1.10pt] (\xm,0) -- (\xm,4.05);
    \node[selected, font=\huge] at (\xm,4.36) {$\phi_m$};
\end{scope}

\end{tikzpicture}
    }
    \caption{Phase measurement viewpoint for phase-squeezed PSK.}
    \label{fig:phase_measurement}
\end{figure}
\subsection{Adaptive Mark II Measurements}
\label{sec:full_quantum_markII}

In this work, we adopt the adaptive
Mark~II phase measurement proposed by \cite{Wiseman1998AdaptiveSingleShotPhaseMeasurements}. For the Mark~II scheme, the phase measurement is described by a POM
density $F_{\mathrm{II}}(\phi)$ on $\phi\in[0,2\pi)$.  For any such unbiased phase measurement, this POM can be written in the
number basis as
\begin{equation}
  F_{\mathrm{II}}(\phi)
  =
  \frac{1}{2\pi}
  \sum_{m,n\ge 0}
  e^{i(m-n)\phi}\,
  H^{(\mathrm{II})}_{mn}\,
  |m\rangle\langle n|,
  \label{eq:F_markII_H}
\end{equation}
where $H^{(\mathrm{II})}$ is a positive Hermitian matrix with
$H^{(\mathrm{II})}_{nn}=1$. Thus, the Mark~II receiver is fully
characterized by the matrix $H^{(\mathrm{II})}$. A central result of \cite{Wiseman1998AdaptiveSingleShotPhaseMeasurements} is that the exact Mark~II matrix
elements are given by
\begin{equation}
  H^{(\mathrm{II})}_{mn}
  =
  \sum_{p=0}^{\lfloor m/2 \rfloor}
  \sum_{q=0}^{\lfloor n/2 \rfloor}
  \gamma_{mp}\,\gamma_{nq}\,
  \left\langle
    \left(\frac{1+C}{1+C^*}\right)^{\frac{n-m}{2}}
    C^p (C^*)^q
  \right\rangle_Q ,
  \label{eq:H_markII_exact}
\end{equation}
with
\begin{equation}
  \gamma_{mp}
  =
  \sqrt{\frac{m!}{2^p (m-2p)! p!}} .
  \label{eq:gamma_mp}
\end{equation}
Here $C$ is one of the sufficient statistics of the adaptive
measurement, and $\langle\cdot\rangle_Q$ denotes averaging over the
ostensible process introduced in the paper. Importantly, for the numerical calculation one does \emph{not} need to choose
$C$ and $C^*$ directly. Instead, \eqref{eq:H_markII_exact} is
evaluated through their ostensible moments
\begin{equation}
  M_{n,m}
  \equiv
  \langle C^n (C^*)^m\rangle_Q ,
  \label{eq:Mnm_def}
\end{equation}
which are computed recursively as follows
\begin{equation}
  M_{n,m}
  =
  \frac{n\,M_{n-1,m}+m\,M_{n,m-1}}
       {2(n-m)^2+n+m},
  \label{eq:Mnm_recurrence}
\end{equation}
with boundary values
\begin{equation}
  M_{n,0}=M_{0,n}=\frac{1}{(2n+1)!!}.
  \label{eq:Mnm_boundary}
\end{equation}
Thus, the practical route to the exact Mark-II matrix is to first compute the moment table $M_{n,m}$ from
\eqref{eq:Mnm_recurrence}-\eqref{eq:Mnm_boundary}, and  then substitute
those moments into \eqref{eq:H_markII_exact} after replacing the factor
\(\big[(1+C)/(1+C^*)\big]^{-m/2}\) by its \(K\)-th order Maclaurin
polynomial in \(C\) and \(C^*\).  An approximation for Mark II matrix was also given in \cite{Wiseman1998AdaptiveSingleShotPhaseMeasurements}, as follows
\begin{equation}
  H^{(\mathrm{II})}_{n+1,n}
  =
  1-h_{II}(n),
\end{equation}
where the Mark~II correction obeys
\begin{equation}
  h_{II}(n)
  \simeq
  \frac{1}{16\,n^{3/2}},
  \qquad n\to\infty .
  \label{eq:hII_asymptotic}
\end{equation}
For further details on the derivation and the intuition behind the Mark II receiver, the reader is referred to Section~III,~IV, and Appendix~A of \cite{Wiseman1998AdaptiveSingleShotPhaseMeasurements}.

\subsection{Mark-II Phase Distribution for Displaced-Squeezed PSK}
\label{subsec:exact_phase_PDF}

We now apply the Mark-II phase POM of
\eqref{eq:F_markII_H} to the phase-squeezed PSK symbols considered in this
work, to obtain the PDF of the measured phase. Let $  |\psi_m\rangle  = \sum_{n=0}^{\infty} c_n^{(m)} |n\rangle
  \label{eq:psi_m_fock_expansion}$
denote the Fock expansion of the \(m\)-th PSK symbol given by (13), where
\(c_n^{(m)}=\langle n|\psi_m\rangle\). The phase probability density associated with
this symbol is $p_m(\phi) =
  \mathrm{Tr}\!\big[
    \rho_m F_{\mathrm{II}}(\phi)
  \big]$. Using  \eqref{eq:F_markII_H} we obtain
\begin{align}
  p_m^{(\mathrm{II})}(\phi)
    &=
    \mathrm{Tr}\!\big[
      |\psi_m\rangle\langle\psi_m|\,
      F_{\mathrm{II}}(\phi)
    \big] \nonumber\\
    &=
    \frac{1}{2\pi}
  \sum_{k,\ell\ge 0}
      H^{(\mathrm{II})}_{\ell k}\,
      c_k^{(m)} c_\ell^{(m)\,*}\,
      e^{i(\ell-k)\phi}.
  \label{eq:p_m_markII_exact_alt}
\end{align}
Equation~\eqref{eq:p_m_markII_exact_alt} has a useful interpretation. The
displaced-squeezed symbol contributes the Fock-basis coherences
\(c_\ell^{(m)}c_k^{(m)\,*}\), while the Mark-II receiver contributes the
measurement factor \(H^{(\mathrm{II})}_{k\ell}\). Each pair of number states
\((k,\ell)\) therefore contributes a phase harmonic
\(e^{i(k-\ell)\phi}\), weighted jointly by the transmitted state and by the
Mark-II measurement matrix. In this sense, the exact Mark-II phase PDF is an
\(H^{(\mathrm{II})}\)-weighted Fourier series in the measured phase
\(\phi\).

For our numerical evaluations, the infinite Fock-basis sum in
\eqref{eq:p_m_markII_exact_alt} is truncated to a finite cutoff
\( n_{\max}\), yielding
\begin{equation}
  p_{m,n_{\max}}^{(\mathrm{II})}(\phi)
  =
  \frac{1}{2\pi}
  \sum_{k,\ell=0}^{n_{\max}}
      H^{(\mathrm{II})}_{\ell k}\,
      c_k^{(m)} c_\ell^{(m)\,*}\,
      e^{i(\ell-k)\phi}.
  \label{eq:p_m_markII_truncated}
\end{equation}
The cutoff \(n_{\max}\) is chosen large enough that the omitted Fock
components have negligible probability weight and negligible effect on the
resulting SEP. In practice, \(n_{\max}\) is increased
until the computed phase density and SEP are stable under further increases
of the cutoff. Overall, to calculate the Mark-II PDF,  first, we compute the displaced-squeezed Fock
coefficients \(c_n^{(m)}\), then we construct the Mark-II matrix
\(H^{(\mathrm{II})}\), and finally we evaluate the truncated Fourier series in
\eqref{eq:p_m_markII_truncated} and integrate it over the PSK decision
regions.

\subsection{Exact Mark~II SEP for Squeezed $M$-PSK}
\label{subsec:exact_ber_markII}
For $M$-PSK, correct symbol detection is performed only when the received symbol is inside the elliptical wedge defined by the decision regions. These regions are shown in Fig. \ref{fig:psk}. Mathematically, this correct decision region for symbol $m$ is the phase
wedge
\begin{equation}
  \mathcal W_m
  =
  [\phi_{m,-},\phi_{m,+}),
  \qquad
  \phi_{m,\pm}
  =
  \phi_m \pm \frac{\pi}{M},
  \label{eq:wedge_def_markII}
\end{equation}
with angles calculated modulo $2\pi$. Then, the correct-detection probability for symbol $m$ is the total probability mass inside the corresponding wedge, i.e.,
\begin{equation}
  P_{c,m}^{(\mathrm{II})}
  =
  \int_{\mathcal W_m} p_{m,n_{\max}}^{(\mathrm{II})}(\phi)\,d\phi ,
  \label{eq:Pc_m_exact_markII}
\end{equation}
and the exact symbol error probability is therefore
\begin{equation}
  P_{e,m}^{(\mathrm{II})}
  =
  1 - P_{c,m}^{(\mathrm{II})}
  =
  1 - \int_{\mathcal W_m} p_{m,n_{\max}}^{(\mathrm{II})}(\phi)\,d\phi .
  \label{eq:Pe_m_exact_markII}
\end{equation}
Substituting $p_{m,n_{\max}}^{(\mathrm{II})}(\phi)$ into
\eqref{eq:Pc_m_exact_markII} gives
\begin{equation}
  P_{c,m}^{(\mathrm{II})}
  =
  \frac{1}{2\pi}
  \sum_{k,\ell=0}^{n_{\max}}
  H^{(\mathrm{II})}_{k\ell}\,
  c_k^{(m)} c_{\ell}^{(m)\,*}
  \int_{\phi_{m,-}}^{\phi_{m,+}}
  e^{i(k-\ell)\phi}\,d\phi .
  \label{eq:Pc_m_exact_series}
\end{equation}
The remaining integral is
\begin{equation}
  \int_{\phi_{m,-}}^{\phi_{m,+}} e^{i(k-\ell)\phi}\,d\phi
  =
  \frac{2\pi}{M}\,
  e^{i(k-\ell)\phi_m}\,
  \text{sinc}\!\left(\frac{(k-\ell)\pi}{M}\right).
\end{equation}
Hence,
\begin{equation}
\begin{aligned}
&  P_{e,m}^{(\mathrm{II})} =\\
& 
  1
  -
  \frac{1}{M}
  \sum_{k,\ell=0}^{n_{\max}}
  H^{(\mathrm{II})}_{k\ell}\,
  c_k^{(m)} c_{\ell}^{(m)\,*}\,
  e^{i(k-\ell)\phi_m}\,
  \text{sinc}\!\left(\frac{(k-\ell)\pi}{M}\right).
  \label{eq:Pe_m_exact_final}
  \end{aligned}
\end{equation}
Assuming equiprobable symbols, the exact average symbol error
probability of the squeezed $M$-PSK constellation under the Mark~II
receiver is
\begin{equation}
  P_e^{(\mathrm{II})}
  =
  \frac{1}{M}\sum_{m=0}^{M-1} P_{e,m}^{(\mathrm{II})}.
  \label{eq:Pe_exact_constellation_final}
\end{equation}
This quantity is the ``\emph{Quantum exact}'' benchmark used throughout the numerical results of the paper. It incorporates both the proposed phase squeezed PSK, through the Fock coefficients $c_n^{(m)}$, and the exact adaptive phase measurement, through the matrix $H^{(\mathrm{II})}$.
\section{Polar-domain SEP Analysis with Mark-II Phase Error Convolution} \label{subsec:polar_markII_exact}
 \begin{figure}[t]
     \centering
     \resizebox{\columnwidth}{!}{%
         \input{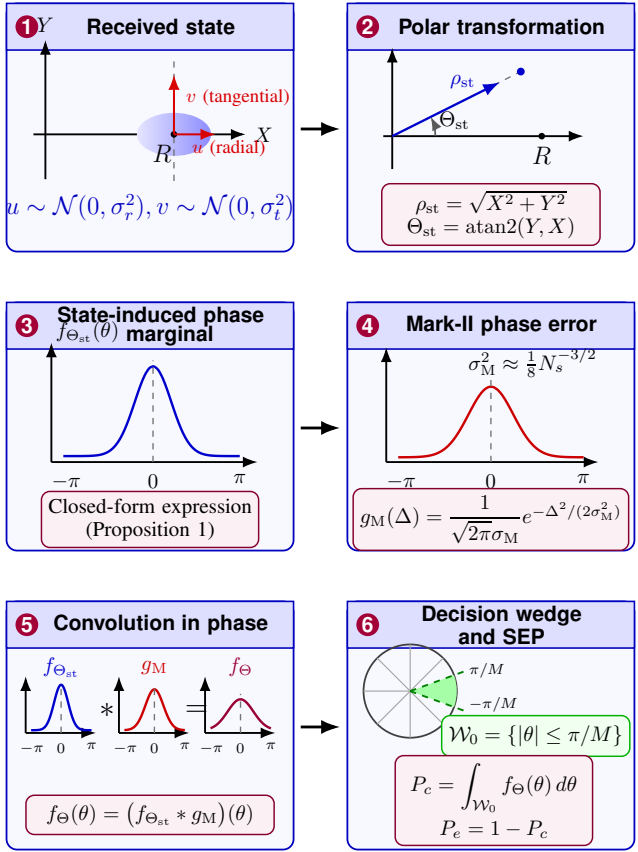}
     }
     \caption{A schematic of the polar-domain SEP analysis.}
     \label{fig:polar analysis}
     \vspace{-0.5cm}
 \end{figure}
While the full quantum Mark~II analysis provides the benchmark error probability, its
evaluation requires the Fock-basis truncation of the transmitted states,
the construction of the Mark~II phase-measurement matrix, and the
numerical evaluation of the resulting Fourier-series expressions. In
particular, the Mark~II matrix itself is obtained through the recursive
calculation of the ostensible moments $M_{n,m}$, by also using the Maclaurin-series
expansion of the factor $\left(\frac{1+C}{1+C^*}\right)^{\frac{n-m}{2}}$. As the photon number and the cutoff dimension increase this exact route becomes computationally cumbersome and
may also suffer from numerical sensitivity issues. It is therefore
useful to derive an alternative error probability formulation, which takes into account the statistics of the squeezed states transmitted and the Mark-II phase noise. As such, in this section we propose a statistical approach for the error probability analysis of the phase-squeezed
$M$-PSK scheme by working in the polar domain. Working in the polar domain is a natural first step, since the error introduced by the Mark-II measurement enters only the phase of the received quantum state.
%The intrinsic state noise and the 
% additional Mark-II receiver phase noise are studied as separate contributions. Then, the analysis proceeds via the joint radius-phase PDF, its marginal densities in $\rho$ and $\theta$, and
% a final convolution in phase with the Mark-II kernel.

\subsection{General Polar-domain SEP Formulation}
In this subsection we provide the step-by-step procedure to calculate the SEP of the squeezed PSK scheme by working in the polar domain. A schematic of the polar domain analysis is given in Fig. \ref{fig:polar analysis}. Without loss of generality, let us assume the
reference PSK symbol to lie on the positive real axis. Equivalently, we take the radius of the PSK symbol to be $R \triangleq |\alpha_0|=\sqrt{N_s}>0.$ Then, we can work in the local Cartesian frame whose radial axis is aligned with the
symbol direction. In this frame, the received noisy symbol in the phase-space is modeled as
\begin{equation}
  Z \triangleq X+iY = (R+u)+iv,
\end{equation}
where $u$ and $v$ denote the \emph{radial} and \emph{tangential} components of
the \emph{state-induced} noise fluctuation.
For the phase-squeezed displaced state we approximate $(u,v)$ as independent
zero-mean Gaussians as follows
\begin{equation}
  \begin{pmatrix}u\\ v\end{pmatrix}
  \sim \mathcal N\!\left(
      \begin{pmatrix}0\\ 0\end{pmatrix},
      \begin{pmatrix}
        \sigma_r^2(r) & 0\\
        0             & \sigma_t^2(r)
      \end{pmatrix}
  \right),
  \label{eq:uv_gaussian_polar}
\end{equation}
with
\begin{equation}
  \sigma_r^2(r)=\frac14 e^{2r},
  \qquad
  \sigma_t^2(r)=\frac14 e^{-2r}.
\end{equation}
We note that the Gaussian modeling of the quantum noise is a usual assumption in the literature \cite{Krikidis2026QuantumRotationDiversity}. The received symbol has Cartesian components $X=R+u$ and $Y=v$,
and therefore its polar coordinates are
\begin{equation}
  \rho_{\mathrm{st}}=\sqrt{X^2+Y^2},
  \qquad
  \Theta_{\mathrm{st}}=\operatorname{arctan2}(Y,X),
\end{equation}
so that $Z=\rho_{\mathrm{st}}e^{j\Theta_{\mathrm{st}}}$. Here, $\operatorname{arctan2}(\cdot,\cdot)$  is the principal angle measure in $[-\pi,\pi]$. It is noted that this interval choice is only a convention and all phases can be interpreted modulo \(2\pi\), so the same density is wrapped onto \([0,2\pi)\) when needed. The joint PDF of the noise components $(u,v)$ under \eqref{eq:uv_gaussian_polar} is
\begin{equation}
  f_{U,V}(u,v)=
  \frac{1}{2\pi\sigma_r\sigma_t}
  \exp\!\left(-\frac{u^2}{2\sigma_r^2}-\frac{v^2}{2\sigma_t^2}\right).
\end{equation}
Using the received-symbol polar transformation
\begin{equation}
  X=\rho\cos\theta,
  \qquad
  Y=\rho\sin\theta,
\end{equation}
together with $u=X-R$ and $v=Y$, we obtain
\begin{equation}
  u=\rho\cos\theta-R,
  \qquad
  v=\rho\sin\theta .
\end{equation}
The Jacobian determinant of the transformation from $(\rho,\theta)$ to $(u,v)$ has magnitude
\begin{equation}
  \left|
  \det
  \begin{pmatrix}
    \partial u/\partial\rho & \partial u/\partial\theta\\[2pt]
    \partial v/\partial\rho & \partial v/\partial\theta
  \end{pmatrix}
  \right|
  =
  \rho,
\end{equation}
and therefore, for $\rho\ge 0$ and $\theta\in(-\pi,\pi]$, the PDF of the received symbol in the polar domain is given by
\begin{align}
  f_{\rho,\Theta_{\mathrm{st}}}(\rho,\theta)
  &=
  f_{U,V}\bigl(\rho\cos\theta-R,\rho\sin\theta\bigr)\,\rho \nonumber\\
  &=
  \frac{\rho}{2\pi\sigma_r\sigma_t}
  \exp\!\left(
    -\frac{(\rho\cos\theta-R)^2}{2\sigma_r^2}
    -\frac{(\rho\sin\theta)^2}{2\sigma_t^2}
  \right),
  \label{eq:joint_rho_theta_state}
\end{align}
with the marginal PDFs given by
\begin{equation}
  f_\rho(\rho)=\int_{-\pi}^{\pi} f_{\rho,\Theta_{\mathrm{st}}}(\rho,\theta)\,d\theta,
  \,\,
  f_{\Theta_{\mathrm{st}}}(\theta)=\int_0^\infty f_{\rho,\Theta_{\mathrm{st}}}(\rho,\theta)\,d\rho.
\end{equation}
Thus, the next step to calculate the SEP, is to calculate the \emph{state-only} phase marginal
$f_{\Theta_{\text{st}}}(\theta)$ in closed form. This is given by the following proposition.
\begin{proposition}
For $ A(\theta)
    \triangleq \frac{\cos^2\theta}{\sigma_r^2}
             +  \frac{\sin^2\theta}{\sigma_t^2}$,   $B(\theta) \triangleq \frac{R\cos\theta}{\sigma_r^2}$, and $C \triangleq \frac{R^2}{\sigma_r^2}$,
the phase marginal $f_{\Theta_{\text{st}}}(\theta)$ is given as
    \begin{equation}
\begin{aligned}
  f_{\Theta_{\text{st}}}(\theta)
   &= \frac{e^{-R^2/(2\sigma_r^2)}}{2\pi\sigma_r\sigma_t}
     \Bigg[
       \frac{1}{A(\theta)} \\
      & + \frac{B(\theta)}{A(\theta)}
         \sqrt{\frac{\pi}{2A(\theta)}}
         \exp\!\left(\frac{B(\theta)^2}{2A(\theta)}\right)
         \erfc\!\left(
           -\frac{B(\theta)}{\sqrt{2A(\theta)}}
         \right)
     \Bigg].
  \label{eq:theta_state_marginal_closed}
\end{aligned}
\end{equation}
\end{proposition}
\begin{IEEEproof}
    The proof is given in Appendix A.
\end{IEEEproof}
This expression can be evaluated for any $\theta$ and
parameters $R,\sigma_r,\sigma_t$, and serves as the starting point
for incorporating the Mark-II phase noise via convolution in the
angular domain. 

In addition, to the phase uncertainty due to the quantum noise, the Mark-II adaptive dyne receiver contributes an additional
\emph{phase} error $\Delta$ on top of the intrinsic state noise. According to \cite{Wiseman1998AdaptiveSingleShotPhaseMeasurements}, for a
mean photon number $N_s = R^2$, the added variance of the phase error scales as
\begin{equation}
  \sigma_{\text{M}}^2(N_s)
    \approx \frac{1}{8}N_s^{-3/2}.
  \label{eq:markII_gaussian_model_polar}
\end{equation}
We note that a well grounded assumption for the Mark-II phase error PDF as is modeling it as a Gaussian PDF \cite{Wiseman1998AdaptiveSingleShotPhaseMeasurements}, thus we have
\begin{equation}
  g_{\text{M}}(\Delta)
   = \frac{1}{\sqrt{2\pi}\,\sigma_{\text{M}}}
     \exp\!\left(-\frac{\Delta^2}{2\sigma_{\text{M}}^2}\right),
   \qquad \Delta\in\mathbb R.
\end{equation}
Also,  $\Theta_{\text{st}}$ and $\Delta$ are independent, while the Mark-II error does not affect the radius, since it introduces only a phase error. Therefore, the total measured phase and radius are given as
\begin{equation}
  \Theta = \Theta_{\text{st}} + \Delta,
  \qquad
  \rho_{\text{tot}} = \rho_{\text{st}}.
\end{equation}

As a consequence, the state-only phase marginal PDF $f_{\Theta_{\text{st}}}$ is mapped to
the total phase marginal PDF $f_{\Theta}$ by a one-dimensional convolution
with the Mark-II error distribution, as follows
\begin{equation}
  f_{\Theta}(\theta)
    = \int_{-\infty}^{+\infty}
        f_{\Theta_{\text{st}}}(\theta-\delta)\,
        g_{\text{M}}(\Delta)\,d\Delta
    = \bigl(f_{\Theta_{\text{st}}} * g_{\text{M}}\bigr)(\theta).
  \label{eq:theta_total_convolution_marginal}
\end{equation}
Because $f_{\Theta_{\text{st}}}$ is integrable and $g_{\text{M}}$ is a
normalized Gaussian, the integral
\eqref{eq:theta_total_convolution_marginal} converges absolutely for
all $\theta$ and defines a smooth, normalized phase PDF. The Gaussian
tails of $g_{\text{M}}$ also allow the truncation of the integration
interval to $[-\Delta_{\max},\Delta_{\max}]$ with a negligible error in the numerical evaluation. We note that finding the convolution of \eqref{eq:theta_total_convolution_marginal} in closed form is not possible due to the complex nature of \eqref{eq:theta_state_marginal_closed}. Moreover, a moment matching approach was also tested, to approximate \eqref{eq:theta_state_marginal_closed}, but failed to produce satisfactory results across all the values $r$ of interest. Nonetheless, the numerical evaluation of \eqref{eq:theta_total_convolution_marginal} is much more computational efficient and stable than the POM analysis of the Mark-II scheme.

% For $M$-PSK, with the assumption that the original symbol lies on the $x$-axis, the nearest-neighbour decision rule in phase declares
% the reference symbol if the measured phase lies in the wedge $\mathcal W_0 = \bigl\{\theta : |\theta|\le \pi/M\bigr\}$.  In terms of the joint PDF
% $f_{\rho,\Theta}$ of the received radius and phase, the correct detection probability is given as
% \begin{equation}
%   P_{\mathrm{c}}(N_s,r)
%     = \int_{0}^{\infty}
%       \int_{\mathcal W_0}
%         f_{\rho_{\text{st}},\Theta_{\text{st}}}(\rho,\theta)\,
%       d\theta\,d\rho.
%   \label{eq:Pc_polar_joint_no_conditionals}
% \end{equation}
% Equivalently, by integrating out the radius, we end up with the marginal-only expression of
% \begin{equation}
%   P_{\mathrm{c}}(N_s,r)
%     = \int_{\mathcal W_0} f_{\Theta_{\text{st}}}(\theta)\,d\theta,
%   \label{eq:Pc_from_phase_marginal_only}
% \end{equation}
% since $\mathcal W_0$ does not depend on $\rho$ and the integration in
% $\rho$ extends over $[0,\infty)$. By symmetry, all $M$ symbols have the same error probability, so the
% average SEP is
% \begin{equation}
%   P_{\mathrm{e}}(N_s,r) = 1 - P_{\mathrm{c}}(N_s,r),
%   \label{eq:Pe_polar_joint_no_conditionals}
% \end{equation}
% In case needed, the calculation of the SEP follows direclty from $P_{\mathrm{e}}(N_s,r)$ using the usual
% Gray-labelling mapping.  

For \(M\)-PSK, the nearest-neighbor phase detection
declares the reference symbol if the measured phase lies in the wedge
\(\mathcal W_0=\{\theta:|\theta|\le \pi/M\}\), with angular quantities
interpreted modulo \(2\pi\). Therefore, after incorporating the Mark-II
phase uncertainty, the correct-detection probability is
\begin{equation}
  P_{\mathrm{c}}(N_s,r)
    =
    \int_{\mathcal W_0} f_{\Theta}(\theta)\,d\theta .
  \label{eq:Pc_from_total_phase_marginal}
\end{equation}
Equivalently, since \(f_{\Theta}=f_{\Theta_{\mathrm{st}}}*g_{\mathrm{M}}\),
this expression can be evaluated by first computing the state-induced
phase marginal in \eqref{eq:theta_state_marginal_closed}, then applying
the angular convolution in
\eqref{eq:theta_total_convolution_marginal}, and finally integrating over
the decision wedge. By symmetry, all \(M\) symbols have the same error
probability, thus,
\begin{equation}
  P_{\mathrm{e}}(N_s,r)=1-P_{\mathrm{c}}(N_s,r).
  \label{eq:Pe_polar_joint_no_conditionals}
\end{equation}
In case needed, the calculation of the bit error rate (BER) of the phase-squeezed PSK follows directly from $P_{\mathrm{e}}(N_s,r)$ using the usual assumption of
the Gray code mapping.  

\section{Closed-form SEP Analysis with  Effective Tangential Variance}
In this section, we get a closed form analysis for the $M$-PSK scheme by considering the effect of the Mark-II phase error in the tangential domain instead, and then calculating the SEP in the polar domain. This way, we fold the Mark-II phase uncertainty into an effective tangential variance and we avoid the need to evaluate the convolution defined in \eqref{eq:theta_total_convolution_marginal}.
\begin{figure}[t]
    \centering
    \resizebox{\columnwidth}{!}{%
        \begin{tikzpicture}[
    x=1cm,y=1cm,
    >=Latex,
    font=\sffamily,
    paneltitle/.style={
        font=\sffamily\bfseries\normalsize,
        align=center
    },
    axis/.style={->, line width=0.95pt, black},
    guide/.style={black!55, dashed, line width=0.70pt},
    phasearc/.style={blue!75!black, dotted, line width=1.30pt},
    statevec/.style={blue!80!black, line width=1.20pt},
    markvec/.style={orange!95!red, line width=1.20pt},
    totalvec/.style={black, line width=1.20pt},
    decisionfill/.style={green!35, fill opacity=0.28},
    small/.style={font=\Large},
    eq/.style={font=\Large}
]

% =========================================================
% Tighter bounding box: controls caption spacing
% =========================================================
\path[use as bounding box] (-0.25,-1.95) rectangle (14.20,4.15);

% =========================================================
% Panel titles
% =========================================================

% =========================================================
% PANEL (a)
% =========================================================
\begin{scope}[shift={(0.35,0)}]

    % -----------------------------------------------------
    % Parameters
    % -----------------------------------------------------
    \def\Rmean{5.00}
    \def\thetaW{30}     % decision wedge half-angle
    \def\thetaM{18}     % Mark-II phase-error angle shown
    \def\yState{-1.35}
    \def\yMark{1.45}

    % -----------------------------------------------------
    % Axes
    % -----------------------------------------------------
    \draw[axis] (-0.25,0) -- (6.45,-0.1) node[anchor=west] {$X$};
    \draw[axis] (0,-1.75) -- (0,3.25) node[anchor=south] {$Y$};

    \node[font=\large, anchor=north east] at (0,0) {$O$};

    % -----------------------------------------------------
    % Decision wedge: separate from variance lengths
    % -----------------------------------------------------
    \fill[decisionfill] (0,0) -- (\thetaW:5.30) -- (-\thetaW:5.30) -- cycle;

    \draw[guide] (0,0) -- (\thetaW:5.25);
    \draw[guide] (0,0) -- (-\thetaW:5.25);

    \node[black!65, font=\small, align=center]
        at (2.20,-0.70)
        {decision\\[-1pt]region};

    % -----------------------------------------------------
    % Dotted phase-circle arc
    % This arc is just the circular phase locus, not a variance length.
    % -----------------------------------------------------
    % \draw[phasearc]
    %     ({-34}:\Rmean) arc[start angle=-34,end angle=34,radius=\Rmean];
    \draw[phasearc]
        ({-27}:\Rmean) arc[start angle=-27,end angle=27,radius=\Rmean];
    % -----------------------------------------------------
    % Mean state and squeezed cloud
    % -----------------------------------------------------
    \draw[black, line width=0.95pt] (\Rmean,0) ellipse (1.15 and 0.23);
    \fill[black] (\Rmean,0) circle (0.065);

    % Radius guide
    % \draw[blue!80!black, dashed, line width=0.75pt] (0,0) -- (\Rmean,0);
    \node[blue!80!black, font=\large, anchor=south]
        at (2.55,0.05) {$R$};

    % Local tangent through the mean
    % \draw[black!45, dashed, line width=0.65pt]
    %     (\Rmean,-1.50) -- (\Rmean,1.60);

    % -----------------------------------------------------
    % Mark-II phase variance on the same circular arc
    % Starts at the mean point (R,0)
    % -----------------------------------------------------
    \draw[->, markvec]
        (0:\Rmean) arc[start angle=0,end angle=\thetaM,radius=\Rmean];

    % Purple radius to the endpoint of the Mark-II phase variance arc
\draw[purple!80!black, dashed, line width=0.95pt]
    (0,0) -- (\thetaM:\Rmean);
    
    \node[orange!95!red, small, anchor=south west]
        at ({\Rmean*cos(\thetaM)-0.8},{\Rmean*sin(\thetaM)-0.9})
        {$\sigma_{\mathrm M}^{2}$};

    % -----------------------------------------------------
    % Tangential approximation of the same phase-error arc
    % Same start point as the arc: the mean point R
    % -----------------------------------------------------
    \draw[->, markvec] (\Rmean,0) -- (\Rmean,\yMark);

    \node[orange!95!red, small, anchor=west]
        at (\Rmean+0.18,0.75)
        {$R^2\sigma_{\mathrm M}^{2}$};

    % Dotted connector from arc endpoint to tangent endpoint
    % \draw[black!65, dotted, line width=1.00pt]
    %     ({\Rmean*cos(\thetaM)},{\Rmean*sin(\thetaM)})
    %     -- (\Rmean,\yMark);

    % -----------------------------------------------------
    % State tangential variance: representative opposite direction
    % -----------------------------------------------------
    \draw[->, statevec] (\Rmean,0) -- (\Rmean,\yState);

    \node[blue!80!black, small, anchor=west]
        at (\Rmean+0.13,-0.82)
        {$\sigma_{t}^2$};

\end{scope}

% =========================================================
% Arrow between panels
% =========================================================
\draw[->, line width=1.35pt, black] (7.00,0.20) -- (7.55,0.20);

% =========================================================
% PANEL (b)
% =========================================================
\begin{scope}[shift={(8.15,0)}]

    % -----------------------------------------------------
    % Parameters
    % -----------------------------------------------------
    \def\Rb{2.55}
    \def\hState{1.10}
    \def\hMark{1.80}
    \def\hEff{2.95}

    % Tiny offsets: visually separated but conceptually collinear
    \def\dxState{-0.06}
    \def\dxMark{0.00}
    \def\dxEff{0.06}

    % -----------------------------------------------------
    % Axis and squeezed cloud
    % -----------------------------------------------------
    \draw[axis] (-0.20,0) -- (5.35,0) node[anchor=west] {$X$};

    \draw[black, line width=0.95pt] (\Rb,0) ellipse (1.25 and 0.22);
    \fill[black] (\Rb,0) circle (0.065);

    % Local tangent guide
    \draw[black!45, dashed, line width=0.65pt]
        (\Rb,-0.25) -- (\Rb,3.15);

    % -----------------------------------------------------
    % Collinear/resultant-style variance arrows
    % all start from the same mean point and lie on the same tangent
    % -----------------------------------------------------
    \draw[->, statevec]
        ({\Rb+\dxState},0) -- ({\Rb+\dxState},\hState);

    \node[blue!80!black, small, anchor=east]
        at ({\Rb-0.18},{0.62*\hState})
        {$\sigma_{t}^2$};

    \draw[->, markvec]
        ({\Rb+\dxMark},0) -- ({\Rb+\dxMark},\hMark);

    \node[orange!95!red, small, anchor=west]
        at ({\Rb+0.14},{0.62*\hMark})
        {$R^2\sigma_{\mathrm M}^{2}$};

    \draw[->, totalvec]
        ({\Rb+\dxEff},0) -- ({\Rb+\dxEff},\hEff);

    \node[black, small, anchor=west]
        at ({\Rb+0.22},{0.83*\hEff})
        {$\sigma_{t,\mathrm{eff}}^2$};

    % Optional height reference markers
    \draw[blue!80!black, dashed, line width=0.70pt]
        ({\Rb-0.34},\hState) -- ({\Rb+0.20},\hState);

    \draw[orange!95!red, dashed, line width=0.70pt]
        ({\Rb-0.34},\hMark) -- ({\Rb+0.20},\hMark);

    \draw[black, dashed, line width=0.70pt]
        ({\Rb-0.34},\hEff) -- ({\Rb+0.20},\hEff);

    % Equation
    \node[eq, anchor=west]
        at (0.5,-1.05)
        {$\sigma_{t,\mathrm{eff}}^2
          =
          \sigma_{t}^2
          +
          R^2\sigma_{\mathrm M}^2$};

\end{scope}

\end{tikzpicture}
    }
    \vspace{0.2cm}
    \caption{The effective variance assumption.}
    \label{fig:tangential}
\end{figure}
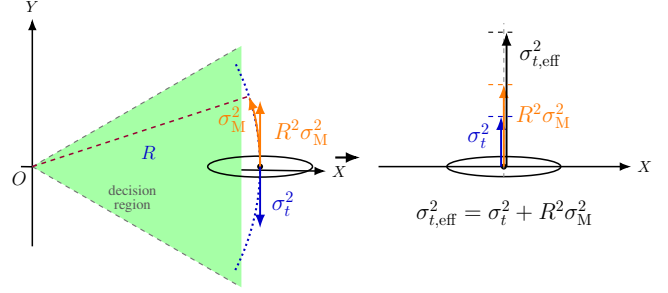
Again, we consider the $M$-PSK symbol to be placed on the positive real axis and $X = R + u$ and $Y = v$, with $u \sim \mathcal{N}\!\bigl(0,\sigma_r^2(r,N_s)\bigr)$, and $v \sim \mathcal{N}\!\bigl(0,\sigma_{t,\mathrm{state}}^2(r)\bigr)$ as before. For small angular fluctuations, which is a rational assumption due to squeezing at the tangential direction, a phase error $\Delta$ corresponds to a
tangential displacement $v \approx R \,\Delta$.  Under this small-angle
approximation, the Mark-II contribution can be represented as an additional
tangential Gaussian noise with variance
\begin{equation}
    \sigma_{t,\mathrm{M}}^2(N_s,r)
      \simeq R^2 \sigma^2_M(\delta\phi_{\mathrm{M}}\bigr).
\end{equation}
We therefore define the overall \emph{effective} tangential variance
\begin{equation}
    \sigma_{t,\mathrm{eff}}^2(N_s,r)
      \triangleq \sigma_{t,\mathrm{state}}^2(r) 
                + \sigma_{t,\mathrm{M}}^2(N_s,r).
    \label{eq:sigma_t_eff}
\end{equation}
% and approximate the joint state and receiver noise as
% \begin{equation}
%     u \sim \mathcal{N}\!\bigl(0,\sigma_r^2(r,N_s)\bigr), 
%     \qquad
%     v \sim \mathcal{N}\!\bigl(0,\sigma_{t,\mathrm{eff}}^2(N_s,r)\bigr).
% \end{equation}
% with $u$ and $v$ being independent. The noisy symbol in the Cartesian and the polar domain is then
% \begin{equation}
% \begin{aligned}
%    X = R + u&, \quad Y = v, \\
%   \rho = \sqrt{X^2 + Y^2}&, \quad \phi = \atantwo(Y,X).
%     \end{aligned}
% \end{equation}
The effective variance assumption is depicted in Fig. \ref{fig:tangential}. Since the considered symbol lies on the $x$-axis, the nearest-neighbour detection for $M$-PSK in phase corresponds to the wedge
decision rule $\rho > 0,\,\,|\phi| < \frac{\pi}{M}$. However, differently than the previous analysis, to ensure  that $\rho>0$ we need to ensure that $R+u>0$, thus by using $\phi = \operatorname{arctan2}(v,R+u)$, we obtain the equivalent decision region in the local
$(u,v)$ coordinates
\begin{equation}
    R + u > 0,
    \qquad
    \bigl|\operatorname{arctan2}(v, R+u)\bigr| < \frac{\pi}{M},
\end{equation}
which further reduces to the following wedge boundaries
\begin{equation}
    u > -R,
    \qquad
    |v| < (R+u)\tan\theta_0,
    \label{eq:uv_wedge}
\end{equation}
with $\theta_0 = \frac{\pi}{M}$. Then, the correct-decision probability for the
reference symbol on the positive real axis is
\begin{equation}
\begin{aligned}
   & P_{\mathrm{c}}(r,N_s)\\
&      = \int_{u=-R}^{\infty}
          f_u(u;\sigma_r^2) 
          \left[
            \int_{-(R+u)\tan\theta_0}^{(R+u)\tan\theta_0}
              f_v(v;\sigma_{t,\mathrm{eff}}^2) \, dv
          \right] du,
    \label{eq:Pc_uv_double_integral}
\end{aligned}
\end{equation}
where
\begin{align}
    f_u(u;\sigma_r^2)
      &= \frac{1}{\sqrt{2\pi}\,\sigma_r}
         \exp\!\left(-\frac{u^2}{2\sigma_r^2}\right),\\[0.5ex]
    f_v(v;\sigma_{t,\mathrm{eff}}^2)
      &= \frac{1}{\sqrt{2\pi}\,\sigma_{t,\mathrm{eff}}}
         \exp\!\left(-\frac{v^2}{2\sigma_{t,\mathrm{eff}}^2}\right).
\end{align}
The SEP, $P_{\mathrm{s}}(r,N_s)$, of the proposed phase-squeezed $M$-PSK scheme under the effective tangential variance assumption is given by the following proposition.
\begin{proposition}
Let  $\gamma = \frac{\tan\theta_0}{\sigma_{t,\mathrm{eff}}}\,R$, and $\beta = \frac{\tan\theta_0}{\sigma_{t,\mathrm{eff}}}\,\sigma_r$. Then
\begin{equation}
    P_{\mathrm{s}}(r,N_s)=1-P_{\mathrm{c}}(r,N_s),
\end{equation}
with the correct detection probability given as 
\begin{equation}
  P_{\mathrm{c}}(r,N_s)
  =
  \Phi\!\left(\frac{\gamma}{\sqrt{1+\beta^2}}\right)
  +
  2T\!\left(
    -\frac{\gamma}{\sqrt{1+\beta^2}},
    -\frac{1}{\beta}
  \right),
  \label{eq:Pc_closed_form_simplified}
\end{equation}
where $T\!\left(\cdot\right)$ denotes the Owen's $T$-function \cite{Owen1956}, which for $h,a\in\mathbb{R}$ is defined as
\begin{equation}
  T(h,a)
  = \frac{1}{2\pi}\int_0^a
      \frac{\exp\!\left(-\tfrac{1}{2}h^2(1+t^2)\right)}{1+t^2}\,dt.
  \label{eq:OwenT_def}
\end{equation}
\end{proposition}
\begin{IEEEproof}
The proof is given in Appendix B.
\end{IEEEproof}
Some useful insights from closed form SEP are given below.
\begin{itemize}
  \item Squeezing reshapes the quantum noise rather than uniformly
  reducing it. In the phase-squeezed configuration, the tangential
  variance is reduced while the radial variance is increased. This
  tradeoff appears explicitly in the denominator
  \(\sqrt{\sigma_{t,\mathrm{eff}}^2+\tan^2\!\theta_0\,\sigma_r^2}\).
  Thus, the error probability depends on the joint effect of noise on the
  tangential and radial fluctuations.
  \item The Mark-II contribution enters through the effective tangential
  variance \(\sigma_{t,\mathrm{eff}}^2\). When
  \(\sigma_{t,\mathrm{eff}}^2\) dominates the denominator, improving the
  phase measurement directly increases the normalized separation. When
  \(\tan^2\!\theta_0\,\sigma_r^2\) dominates instead, radial
  anti-squeezing limits the benefit of reducing the Mark-II phase-noise
  contribution.
  \item Under the fixed photon-number constraint, the optimal squeezing
  generally results from a balance between displacement and noise shaping.
  Increasing \(r\) reduces the tangential state noise, but it also
  increases the radial anti-squeezed variance and reduces the displacement
  energy available. The simplified expression therefore supports
  the existence of an intermediate operating regime in which the
  constellation remains sufficiently separated while the phase-sensitive
  noise is sufficiently suppressed.
\end{itemize}
Interestingly, this is the second time in the literature that the Owen's $T$-function appears when calculating the SEP of a constellation. The first, was in \cite{oikonomou2025constellation}, in which the effect of phase noise on the constellation design was studied. This suggests a potential analogy between phase errors in classical communication and phase uncertainty in quantum communication introduced by squeezing and phase measurement. 
\section{Numerical Results}
% This section evaluates the proposed phase-squeezed $M$-PSK modulation under the fixed total-photon constraint of \eqref{eq:Ntot_split_psk} by comparing (i) the full quantum Mark-II receiver model, (ii) the Mark-II quantum approximation, (iii) the proposed polar-domain approximation, (iv) the proposed $T$-Owens approximation and (v) the coherent-transmission case. An expected observation is that creating the matrix $H_{n,m}^{II}$ is getting more expensive as the the PSK modulation order increases, or while the SEP decreases. This is attributed to the fact that the available information is spread out to more distant Fock basis and the truncation order increases accordingly. Moreover, for the matrix to be well behaved and to follow the theory, as the truncation value increases the McLauring terms have also increase. The values we used are $n_\mathrm{cut} = 80, 100, 150, 200$ and for the MacLauring approximation $50, 100, 150, 200$ in respect. Also, whenever needed, $10^6$ Monte Carlo samples were used to extract the SER.
This section evaluates the proposed phase-squeezed $M$-PSK modulation under the fixed total-photon constraint of \eqref{eq:Ntot_split_psk} by comparing (i) the ``Quantum exact" scheme which is calculated by (30) and (31) (ii) the ``\emph{Quantum approximation}" scheme which occurs from calculating SEP by (22) and (31) (iii) the ``\emph{Analysis - polar domain}" scheme which is calculated by (47)-(48) and Proposition 1 (iv) the ``\emph{Analysis $T$-Owens}" which is given by Proposition 2 and (v) the \emph{no squeezing} scheme, which is $M$-PSK with coherent states. A practical consideration in the full quantum evaluation is the truncation of the Fock expansion. As $N_{\mathrm{tot}}$ increases and/or as the target error probabilities become smaller, the displaced-squeezed symbol states exhibit non-negligible weight over a wider range of photon numbers. Accurate evaluation of the Mark-II phase distribution, therefore, requires a larger cutoff $n_{\mathrm{max}}$. In our simulations we used $n_{\mathrm{max}}\in\{80,100,150,200\}$ depending on the operating point. When increasing $n_{\mathrm{max}}$, the Maclaurin-series truncation order $K$ also increased accordingly, until we obtained a Hermitian positive semidefinite $H^{\mathrm{II}}$ with unit diagonal. The values for $K$ were in the set $\{50,100,200,300\}$. We note that the theoretical results of Proposition 1 and 2 were validated with $10^6$ Monte Carlo samples, but the evaluation is omitted for clarity, since the aim is to compare the proposed SEP analysis with the POM based analysis. 

\begin{figure}[t]
    \centering
    \includegraphics[width=\linewidth]{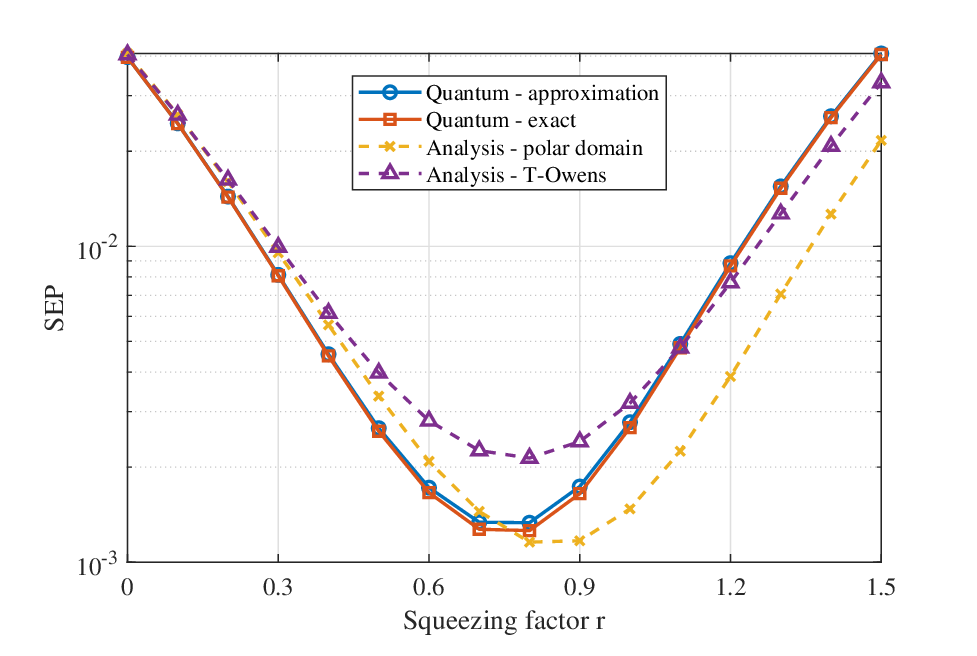}
    \caption{SEP vs the squeezing factor $r$, $M = 16,\, N_\mathrm{tot}=30$.}
    \label{fig:servsr}
\end{figure}

\begin{figure}[t]
    \centering
    \includegraphics[width=.9\linewidth]{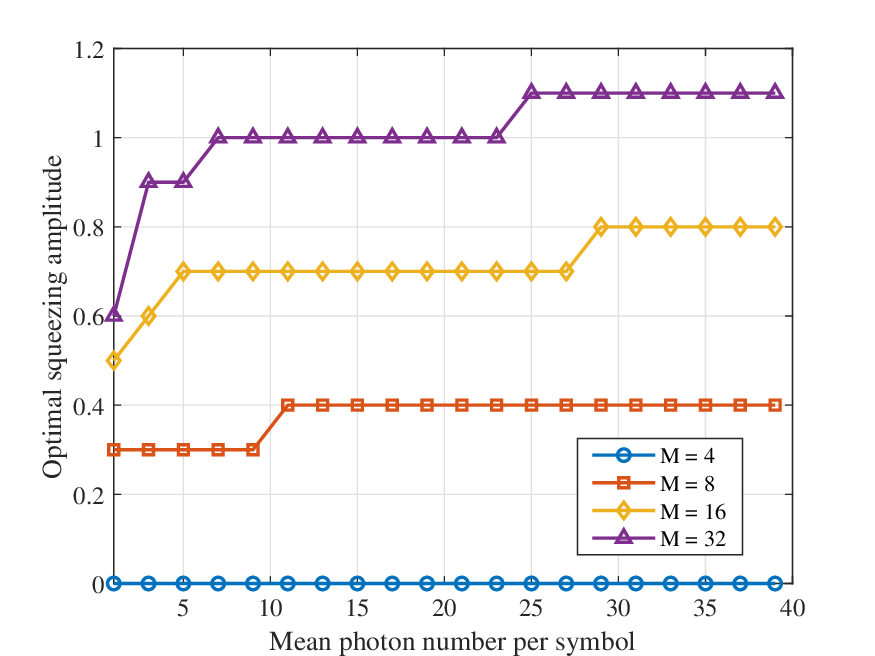}
    \caption{Optimal squeezing factor $r^*$ vs $N_\mathrm{tot}$.}
    \label{fig:vsr}
    \vspace{-0.15cm}
\end{figure}

% Fig.~\ref{fig:servsr} plots the SEP versus the squeezing magnitude $r$ for $M=16$ and $N_{\mathrm{tot}}=30$. Due to the power constraint of \eqref{eq:Ntot_split_psk}, the curve is non-monotone, exhibiting a clear minimum. For small $r$, increasing squeezing reduces the effective angular spread of the quantum state around each constellation point and therefore reduces the overlap between different points. But, beyond a certain point, further squeezing significantly reduces the modulation photons $N_s$, and hence the radius $R$, weakening the separation of neighbouring symbols and increasing errors. This observation explains why both the proposed polar-domain analysis and the $T$-Owens approximation have to consider as error not only the received phase crossing the wedge boundary, but also the radius to be greater than zero.
Fig.~\ref{fig:servsr} shows the SEP as a function of the squeezing magnitude \(r\) for \(M=16\) and \(N_{\mathrm{tot}}=30\). Notably, for this setup SEP drops from $0.04$ to roughly $0.0015$, over an order of magnitude. Because of the power constraint in \eqref{eq:Ntot_split_psk}, the SEP varies non-monotonically and the curve exhibits a clear minimum. When \(r\) is small, increasing the squeezing reduces the angular variance of the received quantum state around each constellation point, which in turn lowers the overlap between neighboring symbols. As \(r\) increases further, however, squeezing consumes a larger fraction of the fixed photon budget, thereby reducing the modulation photon number \(N_s\) and shrinking the constellation radius \(R\). As a result, adjacent symbols become less separated and the SEP starts to rise again. This tradeoff is also reflected in the error analysis. In both the polar-domain treatment and the $T$-Owen's  approximation, errors arise not only from phase excursions across the phase decision boundaries, but also from samples lying in the opposite radial half-plane. It is worth noting that all schemes follow the same non-monotonic SEP trend, and predict similar if not identical optimal squeezing values, suggesting that the proposed approximations capture the SEP behavior accurately over the range considered. Moreover, the polar-domain scheme appears to provide the closer fit. This is because it is based on the convolution in \eqref{eq:theta_total_convolution_marginal}. By contrast, the $T$-Owen's  scheme combines the two variances in \eqref{eq:sigma_t_eff}, which is effectively analogous to assuming the convolution of two Gaussian random variables, which is not necessarily true for all values of \(r\) and \(N_s\).
   
\begin{figure}[!t]
    \centering
    \begin{subfigure}{\columnwidth}
        \centering
        \includegraphics[width=\linewidth]{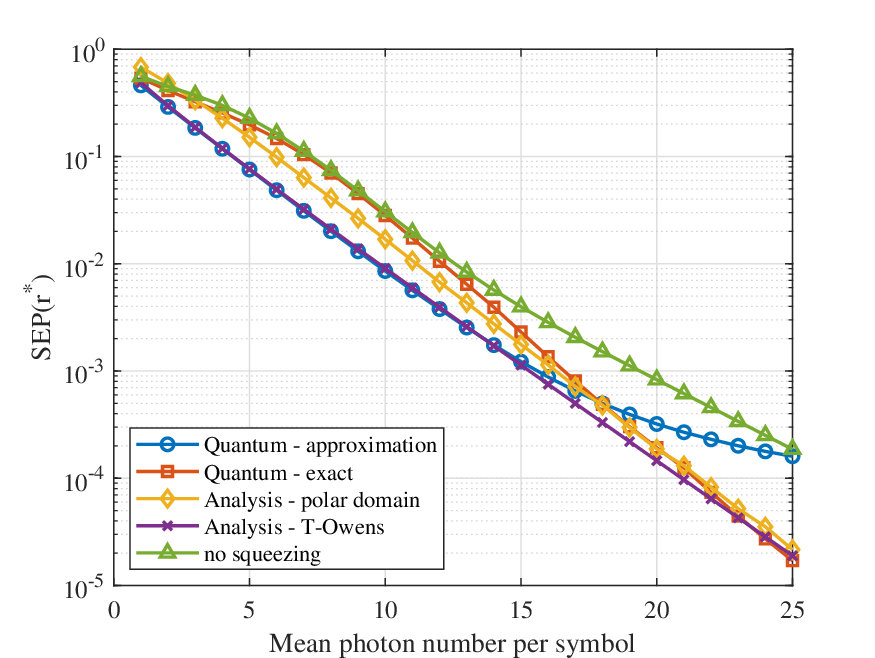}
        \caption{SEP vs $N_\mathrm{tot}$ for $M = 8$.}
        \label{fig:sub1}
    \end{subfigure}
    \begin{subfigure}{\columnwidth}
        \centering
        \includegraphics[width=\linewidth]{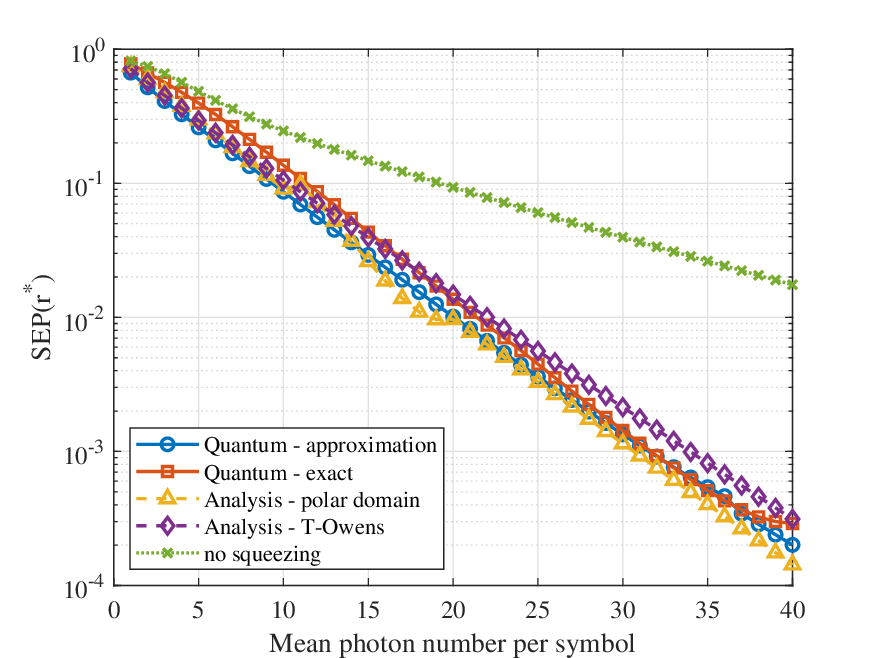}
        \caption{SEP vs $N_\mathrm{tot}$ for $M = 16$.}
        \label{fig:sub2}
    \end{subfigure}
    \begin{subfigure}{\columnwidth}
        \centering
        \includegraphics[width=\linewidth]{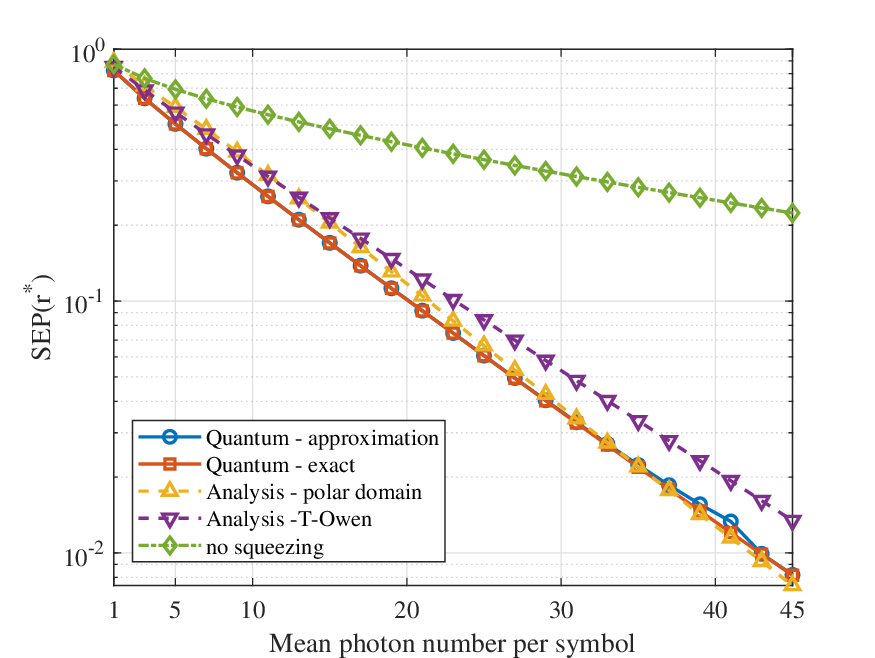}
        \caption{SEP vs $N_\mathrm{tot}$ for $M = 32$.}
        \label{fig:sub3}
    \end{subfigure}
    \caption{SEP vs $N_\mathrm{tot}$ for different modulation orders.}
    \label{fig:three_vertical}
\end{figure}

% Fig.~\ref{fig:vsr} show the optimal $r^*$ versus $N_{\mathrm{tot}}$) and summarizes the same balance in the form of an optimizer trajectory. To gain the optimal squeezing value the SEP based on the full quantum Mark-II scheme was calculated. As $N_{\mathrm{tot}}$ increases, the optimal $r^*$ remains in an intermediate regime rather than growing without bound. While larger $N_{\mathrm{tot}}$ allows more squeezing without collapsing $N_s$, excessive squeezing still diverts energy away from displacement and eventually reduces the effect of squeezing. That's why, from a point on, optimal squeezing seems to remain constant, since further squeezing probably provides no additional gain. Also, it verifies that our analysis complements the analysis of \cite{Djordjevic} since for a QPSK we also get $r^*=0$, which is rational since the optimal noise distribution for a QPSK is that of a circular noise. What we also note is that the optimal values of squeezing also increase as the order of the constellation increases, since the error is mainly defined by the angular distance between symbols and not their radial distance. That's why for the same $N_\mathrm{tot}$ higher order constellations have a greater optimal squeezing value. 
Fig.~\ref{fig:vsr} shows the optimal squeezing value \(r^*\) as a function of \(N_{\mathrm{tot}}\). To obtain \(r^*\), the SEP was evaluated using the full quantum Mark-II scheme, via an exhaustive search with a step of 0.1. As \(N_{\mathrm{tot}}\) increases, the optimal squeezing remains in an intermediate regime rather than increasing without bound. Although a larger \(N_{\mathrm{tot}}\) allows more squeezing without severely reducing \(N_s\), excessive squeezing still diverts energy away from the displacement and eventually weakens the overall benefit. This explains why, beyond a certain point, the optimal squeezing appears to saturate, as further squeezing provides little or no additional gain. The figure also confirms that our analysis is consistent with that of \cite{Bhadani2022OptimizedSqueezingPSKQSD}, since for QPSK we likewise obtain \(r^*=0\). This is physically reasonable, as the optimal noise distribution for QPSK is circularly symmetric. We also observe that the optimal squeezing increases with the constellation order. This is because, for higher-order constellations, the error is determined more strongly by the angular separation between neighboring symbols than by their radial separation. As a result, for the same \(N_{\mathrm{tot}}\), higher-order constellations favor a larger optimal squeezing value.

% Fig.~\ref{fig:three_vertical} reports SEP versus $N_{\mathrm{tot}}$ for $M\in\{8,16,32\}$. First, we note that for all $N_{\mathrm{tot}}$ the proposed approximation schemes appear consistent with the quantum analysis, since they predict a SEP within accuracy of 2-4 photons in general. Notably, for $M=8$, the proposed polar-domain scheme is more accurate that then approximate quantum model of \textcolor{blue}{[ref]}. Again, we see that the proposed polar-domain model is a better fit that the proposed Owens-$T$ model. The polar model is always within accuracy of 2 photons, while for $M=32$ the Owens-$T$ model is within 2-4 photons. That showcases the trade-off between accruacy and complexity. Both models are nonetheless considered consistent. Furthermore, in all cases, SEP decreases rapidly with $N_{\mathrm{tot}}$ and the gain from optimizing $r$ becomes more pronounced as $M$ increases, because the correct detection half-angle $\theta_0=\pi/M$ shrinks and phase concentration becomes increasingly more important. Another notable observatino is the gain in number of photons provided by squeezing. Specifically for $M=8$ we observe a slow start, however eventually we end up with a gain of $5$ photons, a $20\%$ reduction in terms of necessary photons. For $M=16$ we end up with an improvement of $20$ photons, a $50\%$ reduction, while for $M = 32$ we end up with an improvement of $\approx66\%$. This is due to the coherent modulation PSK reaching a floor faster, which is not the case for squeezed PSK, due to the coding gain provided by squeezing. 

Fig.~\ref{fig:three_vertical} reports the SEP as a function of \(N_{\mathrm{tot}}\) for \(M\in\{8,16,32\}\). We note that for each $N_\mathrm{tot}$, the optimal squeezing value was calculated based on an exhaustive search. First, we observe that, for all values of \(N_{\mathrm{tot}}\), the proposed approximation schemes remain consistent with the full quantum analysis, as they predict the SEP with an accuracy corresponding in general to about \(2\) to \(4\) photons. Notably, for \(M=8\), the proposed schemes are more accurate than the approximate quantum model, with the polar-domain model providing a closer fit than the $T$-Owen's model. The polar-domain model remains within an accuracy of about \(2\) photons, whereas for \(M=32\) the $T$-Owen's  model is typically within \(2\) to \(4\) photons. This highlights the tradeoff between accuracy and complexity. Even so, both approximation schemes can be regarded as consistent with the full quantum treatment. Furthermore, in all cases the SEP decreases rapidly with increasing \(N_{\mathrm{tot}}\), and the benefit of optimizing \(r\) becomes more pronounced as \(M\) increases. This is because the correct detection half-angle, \(\theta_0=\pi/M\), becomes smaller for higher-order constellations, making phase concentration increasingly important. Another notable observation is the photon-efficiency gain provided by squeezing. For \(M=8\), the gain is modest at first, but eventually reaches about \(5\) photons, corresponding to roughly a \(20\%\) reduction in the required photon number. For \(M=16\), the improvement reaches about \(20\) photons, corresponding to nearly a \(50\%\) reduction, while for \(M=32\) the photon reduction is approximately \(66\%\). This behavior arises because the conventional quantum PSK with coherent modulation reaches its error floor faster, whereas the proposed PSK with phase squeezed continues to benefit from the SNR gain caused by squeezing.

\begin{figure}[t]
    \centering
    \begin{subfigure}{\columnwidth}
        \centering
        \includegraphics[width=\linewidth]{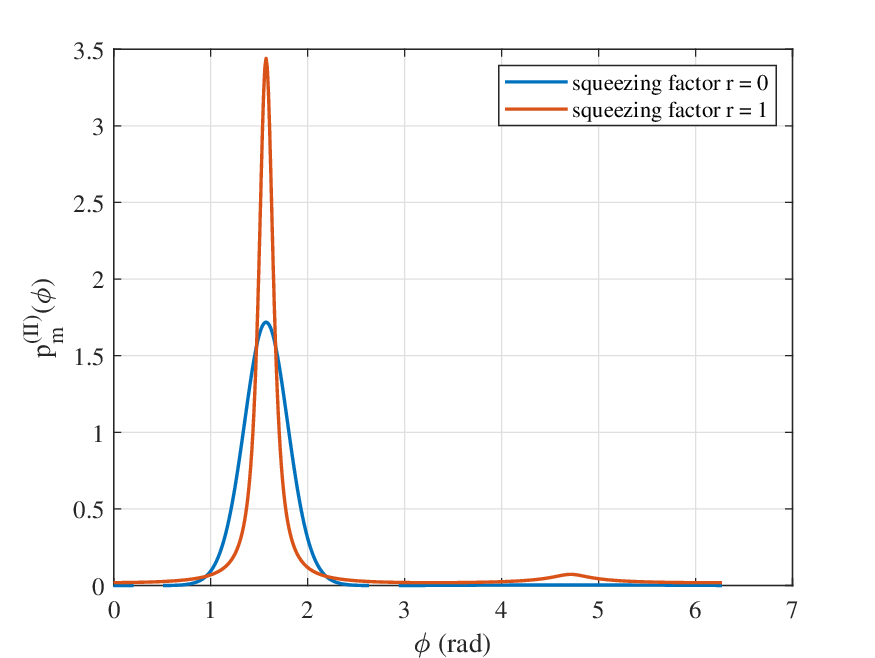}
        \caption{The phase distribution, $N_\mathrm{tot}=5$.}
        \label{fig:pdf1}
    \end{subfigure}
    \begin{subfigure}{\columnwidth}
        \centering
        \includegraphics[width=\linewidth]{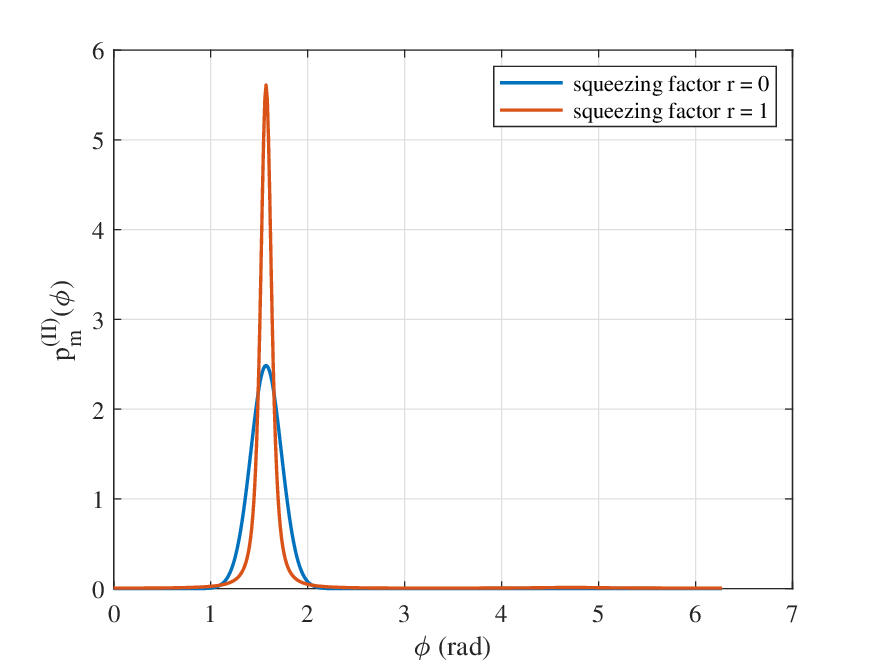}
        \caption{The phase distribution, $N_\mathrm{tot}=10$.}
        \label{fig:pdf2}
    \end{subfigure}
    \caption{The quantum state phase \& Mark-II phase noise distribution.}
    \label{fig:pdfs}
    \vspace{-.5cm}
\end{figure}

% Finally, Fig.~\ref{fig:pdfs} illustrates the symbols phase PDF modulo $2\pi$ based on $\eqref{eq:p_m_exact_fourier_TBD}$, for $M = 16$. For both $N=5,10$ the PDF resembles that of a normal distribution around the modulation phase of the symbol. Nevertheless, this is not the case $\forall r,N_s$ as we can immediately infer from  \eqref{eq: theta_state_marginal_closed}. However, for moderate values of $r$, and while $N_s$ increases the overall PDF could perhaps be well approximated using a normal distriubiton and a proper moment matching. Moreover, we observe that the PDF becomes more concentrated when both $r$ and $N_s$ increase. That's first due to squeezing reducing the phase uncertainty, and increased $N_s$ reducing the Mark-II error measurement. These resutls first verify the correctness of the Mark-II measurement for squeezed $M$-PSK. In addition, the observed narrowing of the total phase distribution explains why the SEP curves approach a steep falloff at high $N_{\mathrm{tot}}$, and why the optimum in $r$ shifts only mildly: once the distribution is already tightly concentrated, further squeezing produces diminishing returns unless it preserves $N_s$.
Finally, Fig.~\ref{fig:pdfs} illustrates the phase PDF, modulo \(2\pi\), obtained from \eqref{eq:p_m_markII_truncated}, for \(M=16\). For both \(N=5\) and \(N=10\), the PDF is centered around the symbol modulation phase and, near its main lobe, resembles a normal distribution. However, this behavior does not hold for all \(r\) and \(N_s\), as can be seen directly from \eqref{eq:theta_state_marginal_closed}. In general, the exact phase PDF is not Gaussian. Still, for moderate values of \(r\) and sufficiently large \(N_s\), the overall PDF could be approximated by a normal distribution after a suitable moment matching. We note that this observation could perhaps aid to simplify the polar-domain analysis. We also observe that the PDF becomes more concentrated as both \(r\) and \(N_s\) increase. This is because squeezing reduces the phase uncertainty of the quantum state, while a larger \(N_s\) reduces the Mark-II measurement error. Normally, the SEP is significantly reduced for higher values of $r$ and $N_s$ also. These results further strengthen the validity of the Mark-II phase measurement model for the squeezed \(M\)-PSK and validate the insights gathered by our theoretical and numerical analysis so far. 

\section{Conclusions}
This paper analyzed the SEP of phase-squeezed \(M\)-PSK under adaptive Mark-II phase measurement and a fixed total photon budget. By assigning symbol-dependent squeezing angles, the proposed constellation aligns the reduced-variance quadrature with the local phase-sensitive direction of each PSK symbol. We derived the POM-based SEP formulation in the Fock basis and complemented it with two statistical approaches of lower complexity, specifically, a polar-domain approximation with an explicit Mark-II angular convolution and a closed-form effective tangential variance approximation involving Owen's \(T\)-function. The results showed that QPSK does not benefit from squeezing, whereas higher-order \(M\)-PSK constellations achieve substantial SEP and photon-efficiency gains, with the improvement increasing with the modulation order. Moreover, the proposed approximations closely tracked the POM analysis thus offering a tractable evaluation alternative. Future work includes extending the analysis to turbulent FSO channels, and generalizing the framework to entanglement-assisted quantum communication.

\appendices

\section{Proof of Proposition 1}
In this appendix,  we provide the proof of Proposition 1. By using the joint polar PDF in~\eqref{eq:joint_rho_theta_state}, the
state-induced phase marginal is given as
\begin{equation}
\begin{aligned}
  f_{\Theta_{\mathrm{st}}}(\theta)
   &=\\
   &\int_{0}^{\infty}
       \!\!\frac{\rho}{2\pi\sigma_r\sigma_t}
       \exp\!\left(
         -\frac{(\rho\cos\theta-R)^2}{2\sigma_r^2}
         -\frac{(\rho\sin\theta)^2}{2\sigma_t^2}
       \right)\,d\rho .
  \label{eq:theta_state_marginal_start}
\end{aligned}
\end{equation}
For each fixed \(\theta\), the exponent is a quadratic polynomial in
\(\rho\). Expanding \((\rho\cos\theta-R)^2\) and collecting the
\(\rho^2\), \(\rho\), and all constant terms we can define
\begin{equation}
  A(\theta)
    \triangleq \frac{\cos^2\theta}{\sigma_r^2}
             +  \frac{\sin^2\theta}{\sigma_t^2},
  \qquad
  B(\theta) \triangleq \frac{R\cos\theta}{\sigma_r^2},
  \qquad
  C \triangleq \frac{R^2}{\sigma_r^2}.
  \label{eq:A_B_C_def}
\end{equation}
Then, we have the following
\begin{equation}
  -\frac{(\rho\cos\theta-R)^2}{2\sigma_r^2}
  -\frac{(\rho\sin\theta)^2}{2\sigma_t^2}
  =
  -\frac{A(\theta)}{2}\rho^2+B(\theta)\rho-\frac{C}{2}.
\end{equation}
Therefore,
\begin{equation}
  f_{\Theta_{\mathrm{st}}}(\theta)
   =
   \frac{e^{-C/2}}{2\pi\sigma_r\sigma_t}
   \int_{0}^{\infty}
       \rho
       \exp\!\left(
         -\frac{A(\theta)}{2}\rho^2+B(\theta)\rho
       \right)d\rho .
  \label{eq:theta_state_marginal_IAB}
\end{equation}
This suggests introducing the auxiliary integral
\begin{equation}
  I(A,B)
    \triangleq
    \int_{0}^{\infty}
      \rho
      \exp\!\left(
        -\frac{A}{2}\rho^2+B\rho
      \right)d\rho,
  \qquad A>0.
  \label{eq:I_AB_def}
\end{equation}
Then, it remains to evaluate \(I(A,B)\). Completing the square term inside the exponential of \eqref{eq:I_AB_def} gives
\begin{equation}
  -\frac{A}{2}\rho^2+B\rho
  =
  -\frac{A}{2}\left(\rho-\frac{B}{A}\right)^2
  +
  \frac{B^2}{2A}.
\end{equation}
Thus,
\begin{equation}
  I(A,B)
  =
  e^{B^2/(2A)}
  \int_{0}^{\infty}
  \rho
  \exp\!\left[
    -\frac{A}{2}\left(\rho-\frac{B}{A}\right)^2
  \right]d\rho .
\end{equation}
Using the change of variables \(x=\rho-B/A\), we obtain
\begin{equation}
  I(A,B)
  =
  e^{B^2/(2A)}
  \int_{-B/A}^{\infty}
  \left(x+\frac{B}{A}\right)e^{-Ax^2/2}\,dx,
\end{equation}
while by splitting the integral into two parts we get
\begin{equation}
\begin{aligned}
  I(A,B)
  &=
  e^{B^2/(2A)}
  \int_{-B/A}^{\infty}x e^{-Ax^2/2}\,dx  \\
  &\quad
  +\frac{B}{A}e^{B^2/(2A)}
  \int_{-B/A}^{\infty}e^{-Ax^2/2}\,dx .
\end{aligned}
\end{equation}
The first integral can be evaluated directly as follows
\begin{equation}
  \int_{-B/A}^{\infty}x e^{-Ax^2/2}\,dx
  =
  \left[-\frac{1}{A}e^{-Ax^2/2}\right]_{-B/A}^{\infty}
  =
  \frac{1}{A}e^{-B^2/(2A)}.
\end{equation}
For the second integral, if we set \(y=\sqrt{A/2}\,x\), we have
\begin{equation}
\begin{aligned}
  \int_{-B/A}^{\infty}e^{-Ax^2/2}\,dx
  &=
  \sqrt{\frac{2}{A}}
  \int_{-B/\sqrt{2A}}^{\infty}e^{-y^2}\,dy  \\
  &=
  \sqrt{\frac{\pi}{2A}}\,
  \erfc\!\left(-\frac{B}{\sqrt{2A}}\right).
\end{aligned}
\end{equation}
Combining these two terms gives the following
\begin{equation}
  I(A,B)
  =
  \frac{1}{A}
  +
  \frac{B}{A}\sqrt{\frac{\pi}{2A}}
  \exp\!\left(\frac{B^2}{2A}\right)
  \erfc\!\left(-\frac{B}{\sqrt{2A}}\right),
  \label{eq:I_AB_closed_form}
\end{equation}
which is valid for all \(A>0\). Finally, using
\(A=A(\theta)\), \(B=B(\theta)\), and \(C=R^2/\sigma_r^2\)  and substituting 
\eqref{eq:I_AB_closed_form} into (68) completes the proof.
\section{Proof of Proposition~2}
Here, we provide the proof of Proposition 2. First, we provide the following auxiliary lemma, whose proof is given in Appendix~C, which will be used to 
evaluate the SEP in terms of Owen's \(T\)-function.

\begin{lemma}
Let
\begin{equation}
  I(x_0,\gamma,\beta)
  \triangleq
  \int_{x_0}^{\infty}
    \phi(x)\Phi(\gamma+\beta x)\,dx,
  \label{eq:lemma_I_general}
\end{equation}
where \(\phi(\cdot)\) and \(\Phi(\cdot)\) denote the standard normal PDF
and CDF, respectively. The, for \(x_0\gamma\neq 0\), it holds
\begin{equation}
\begin{aligned}
  I(x_0,\gamma,\beta) 
  &=\\ 
&  \frac{1}{2}
  \Phi\!\left(\frac{\gamma}{\sqrt{1+\beta^2}}\right)
  -\frac{1}{2}\Phi(x_0)
  +T\!\left(x_0,\frac{\gamma+\beta x_0}{x_0}\right) \\
  +T\!&\left(
      -\frac{\gamma}{\sqrt{1+\beta^2}},
      \frac{\gamma\beta+x_0(1+\beta^2)}{\gamma}
    \right)
  +\frac{1}{2}\mathbf{1}_{\{\gamma/x_0<0\}},
  \label{eq:lemma_I_general_closed}
\end{aligned}
\end{equation}
where \(T(\cdot,\cdot)\) denotes Owen's \(T\)-function.
\label{lem:truncated_gaussian_cdf}
\end{lemma}
\begin{IEEEproof}
    The proof is given in Appendix C.
\end{IEEEproof}

We now proceed with the calculation of \eqref{eq:Pc_uv_double_integral}. It is easy to show that the inner integral of
\eqref{eq:Pc_uv_double_integral} over \(v\) is given as
\begin{equation}
\begin{aligned}
  &\int_{-(R+u)\tan\theta_0}^{(R+u)\tan\theta_0}
        f_v(v;\sigma_{t,\mathrm{eff}}^2)\,dv  =
  \erf\!\left(
        \frac{(R+u)\tan\theta_0}
        {\sigma_{t,\mathrm{eff}}\sqrt{2}}
      \right).
\end{aligned}
\label{eq:v_inner_erf}
\end{equation}
Then, substituting \eqref{eq:v_inner_erf} into
\eqref{eq:Pc_uv_double_integral}, we obtain
\begin{equation}
\begin{aligned}
  &P_{\mathrm{c}}(r,N_s)
  =\\
&  \int_{-R}^{\infty}
    \frac{1}{\sqrt{2\pi}\sigma_r}
    \exp\!\left(-\frac{u^2}{2\sigma_r^2}\right)
    \erf\!\left(
      \frac{(R+u)\tan\theta_0}
      {\sigma_{t,\mathrm{eff}}\sqrt{2}}
    \right)du .
\end{aligned}
\label{eq:Pc_uv_1D}
\end{equation}
Using the property \(\erf(z)=2\Phi(\sqrt{2}z)-1\), and defining $a \triangleq \frac{\tan\theta_0}{\sigma_{t,\mathrm{eff}}}$,
the error function in \eqref{eq:Pc_uv_1D} becomes
\begin{equation}
  \erf\!\left(
    \frac{(R+u)\tan\theta_0}
    {\sigma_{t,\mathrm{eff}}\sqrt{2}}
  \right)
  =
  2\Phi\!\left(a(R+u)\right)-1.
\end{equation}
Hence it holds
\begin{equation}
\begin{aligned}
  &P_{\mathrm{c}}(r,N_s)
  =\\
&  2\int_{-R}^{\infty}
      f_u(u;\sigma_r^2)
      \Phi\!\left(a(R+u)\right)du 
  -
  \int_{-R}^{\infty}f_u(u;\sigma_r^2)\,du .
\end{aligned}
\end{equation}
The second term is the Gaussian tail
\begin{equation}
  \int_{-R}^{\infty}f_u(u;\sigma_r^2)\,du
  =
  \Phi\!\left(\frac{R}{\sigma_r}\right).
\end{equation}
Therefore, by defining
\begin{equation}
  I_1
  \triangleq
  \int_{-R}^{\infty}
      f_u(u;\sigma_r^2)
      \Phi\!\left(a(R+u)\right)du,
  \label{eq:I1_original_u}
\end{equation}
we obtain
\begin{equation}
  P_{\mathrm{c}}(r,N_s)
  =
  2I_1-\Phi\!\left(\frac{R}{\sigma_r}\right).
  \label{eq:Pc_from_I1}
\end{equation}

It remains to calculate \(I_1\) by using 
Lemma~\ref{lem:truncated_gaussian_cdf}. Let us set the following
\begin{equation}
  x=\frac{u}{\sigma_r},
  \qquad
  x_0=-\frac{R}{\sigma_r},
  \qquad
  \gamma=aR,
  \qquad
  \beta=a\sigma_r .
\end{equation}
Then, \(f_u(u;\sigma_r^2)\,du=\phi(x)\,dx\), and
\(a(R+u)=\gamma+\beta x\). Thus,
\begin{equation}
  I_1
  =
  \int_{x_0}^{\infty}
    \phi(x)\Phi(\gamma+\beta x)\,dx.
  \label{eq:I1_phiPhi}
\end{equation}
Applying Lemma~\ref{lem:truncated_gaussian_cdf} gives
\begin{equation}
\begin{aligned}
  I_1
  &=
  \frac{1}{2}
  \Phi\!\left(\frac{\gamma}{\sqrt{1+\beta^2}}\right)
  -\frac{1}{2}\Phi(x_0)
  +T\!\left(x_0,\frac{\gamma+\beta x_0}{x_0}\right) \\
  &\quad
  +T\!\left(
      -\frac{\gamma}{\sqrt{1+\beta^2}},
      \frac{\gamma\beta+x_0(1+\beta^2)}{\gamma}
    \right)
  +\frac{1}{2}\mathbf{1}_{\{\gamma/x_0<0\}}.
  \label{eq:I1_from_lemma}
\end{aligned}
\end{equation}
We now connect this general identity to the parameters of the present
problem. Since
\begin{equation}
  x_0=-\frac{R}{\sigma_r},
  \qquad
  \gamma=\frac{R\tan\theta_0}{\sigma_{t,\mathrm{eff}}},
  \qquad
  \beta=\frac{\sigma_r\tan\theta_0}{\sigma_{t,\mathrm{eff}}},
\end{equation}
we have
\begin{equation}
  \gamma+\beta x_0=0,
\end{equation}
and therefore
\begin{equation}
  T\!\left(x_0,\frac{\gamma+\beta x_0}{x_0}\right)
  =
  T(x_0,0)=0.
\end{equation}
Moreover,
\begin{equation}
  \frac{\gamma\beta+x_0(1+\beta^2)}{\gamma}
  =
  -\frac{1}{\beta}.
\end{equation}
Since \(R>0\), \(\sigma_r>0\), \(\sigma_{t,\mathrm{eff}}>0\), and
\(\theta_0>0\), it follows that \(x_0<0\) and \(\gamma>0\). Hence, $\mathbf{1}_{\{\gamma/x_0<0\}}=1$. Substituting these identities into \eqref{eq:I1_from_lemma} yields
\begin{equation}
  I_1
  =
  \frac{1}{2}
  \Phi\!\left(\frac{\gamma}{\sqrt{1+\beta^2}}\right)
  -\frac{1}{2}\Phi(x_0)
  +
  T\!\left(
    -\frac{\gamma}{\sqrt{1+\beta^2}},
    -\frac{1}{\beta}
  \right)
  +\frac{1}{2}.
  \label{eq:I1_final_simplified}
\end{equation}
Finally, substituting \eqref{eq:I1_final_simplified} into
\eqref{eq:Pc_from_I1}, using \(x_0=-R/\sigma_r\), and applying
\(\Phi(z)+\Phi(-z)=1\), the terms
\(-\Phi(x_0)\), \(1\), and \(-\Phi(R/\sigma_r)\) cancel. Thus,
\begin{equation}
  P_{\mathrm{c}}(r,N_s)
  =
  \Phi\!\left(\frac{\gamma}{\sqrt{1+\beta^2}}\right)
  +
  2T\!\left(
    -\frac{\gamma}{\sqrt{1+\beta^2}},
    -\frac{1}{\beta}
  \right).
\end{equation}
Since \(P_{\mathrm{s}}(r,N_s)=1-P_{\mathrm{c}}(r,N_s)\), the proof of
Proposition~2 is complete.

\section{Proof of Lemma~1}

The purpose of this appendix is to prove the auxiliary identity used in
Appendix~B. The integral in \eqref{eq:lemma_I_general} is a truncated
Gaussian integral. The key step is to rewrite the truncated part as a
bivariate normal probability, which can then be evaluated by Owen's
representation of the bivariate normal CDF \cite{Owen1956,Owen1980,Komelj2023}. Starting from \eqref{eq:lemma_I_general}, we write
\begin{equation}
  I(x_0,\gamma,\beta)
  =
  J(\infty)-J(x_0),
  \label{eq:lemma_I_split}
\end{equation}
where
\begin{equation}
  J(\infty)
  =
  \int_{-\infty}^{\infty}
    \phi(x)\Phi(\gamma+\beta x)\,dx,
\end{equation}
and
\begin{equation}
  J(x_0)
  =
  \int_{-\infty}^{x_0}
    \phi(x)\Phi(\gamma+\beta x)\,dx .
\end{equation}
The full-space term is the standard Gaussian identity \cite{GradshteynRyzhik2015TableIntegrals}
\begin{equation}
  J(\infty)
  =
  \Phi\!\left(\frac{\gamma}{\sqrt{1+\beta^2}}\right).
  \label{eq:lemma_J_infty}
\end{equation}

We now evaluate the truncated term \(J(x_0)\). Let
\(X,Z\sim\mathcal N(0,1)\) be two independent random variables. Since
\(\Phi(\gamma+\beta x)=\Pr\{Z\le \gamma+\beta x\}\), we can write
\begin{equation}
  J(x_0)
  =
  \mathbb P\!\left(X\le x_0,\; Z\le \gamma+\beta X\right).
  \label{eq:lemma_Jx0_prob}
\end{equation}
Let us define
\begin{equation}
  Y\triangleq \frac{Z-\beta X}{\sqrt{1+\beta^2}} .
\end{equation}
Then \(Y\) has zero mean and unit variance. Moreover, since \(X\) and
\(Z\) are independent,
\begin{equation}
  \mathrm{Corr}(X,Y)
  =
  \mathbb E[XY]
  =
  -\frac{\beta}{\sqrt{1+\beta^2}}
  \triangleq \tilde{\rho}.
\end{equation}
Thus, \((X,Y)\) is a standard bivariate normal vector with correlation
\(\tilde{\rho}\). The event \(Z\le \gamma+\beta X\) is equivalent to
  $Y\le \frac{\gamma}{\sqrt{1+\beta^2}}$.
Therefore it holds that
\begin{equation}
  J(x_0)
  =
  \Phi_2\!\left(
      x_0,
      \frac{\gamma}{\sqrt{1+\beta^2}};
      -\frac{\beta}{\sqrt{1+\beta^2}}
    \right),
  \label{eq:lemma_Jx0_bivar}
\end{equation}
where \(\Phi_2(\cdot,\cdot;\rho)\) is the standard bivariate normal CDF
with correlation coefficient \(\rho\). We now apply Owen's representation of the bivariate normal CDF \cite{Owen1956,Komelj2023}. For
\((x,y)\neq(0,0)\),
\begin{equation}
\begin{aligned}
  \Phi_2(x,y;\rho)
  &=
  \frac{1}{2}\left(\Phi(x)+\Phi(y)\right) \\
  &
  -T\!\left(x,\frac{y-\rho x}{x\sqrt{1-\rho^2}}\right)
  -T\!\left(y,\frac{x-\rho y}{y\sqrt{1-\rho^2}}\right)
  -\beta_0,
\end{aligned}
\label{eq:Owen_bivar_identity}
\end{equation}
where
\begin{equation}
  \beta_0=
  \begin{cases}
    0, & xy>0 \text{ or } (xy=0 \text{ and } x+y\ge 0),\\[1mm]
    \tfrac{1}{2}, & \text{otherwise}.
  \end{cases}
  \label{eq:beta0_def}
\end{equation}
In \eqref{eq:Owen_bivar_identity}, we set $x=x_0$, $y=\frac{\gamma}{\sqrt{1+\beta^2}}$, and $\rho= -\frac{\beta}{\sqrt{1+\beta^2}}$. Since $\sqrt{1-\tilde{\rho}^2}
  =
  \frac{1}{\sqrt{1+\beta^2}}$,
the first \(T\)-function argument becomes
\begin{equation}
  \frac{y-\tilde{\rho}x}{x\sqrt{1-\tilde{\rho}^2}}
  =
  \frac{\gamma+\beta x_0}{x_0},
  \label{eq:first_T_arg_reduction}
\end{equation}
and the second becomes
\begin{equation}
  \frac{x-\tilde{\rho}y}{y\sqrt{1-\tilde{\rho}^2}}
  =
  \frac{\gamma\beta+x_0(1+\beta^2)}{\gamma}.
  \label{eq:second_T_arg_reduction}
\end{equation}
Thus it holds that
\begin{equation}
\begin{aligned}
  J(x_0)
  =
  \frac{1}{2}\Bigg[
      \Phi(x_0)
      +\Phi\!\left(\frac{\gamma}{\sqrt{1+\beta^2}}\right)
      -2T\!\left(x_0,\frac{\gamma+\beta x_0}{x_0}\right) \\
      -2T\!\left(
          -\frac{\gamma}{\sqrt{1+\beta^2}},
          \frac{\gamma\beta+x_0(1+\beta^2)}{\gamma}
        \right)
      -\mathbf{1}_{\{\gamma/x_0<0\}}
    \Bigg].
  \label{eq:lemma_Jx0_final}
\end{aligned}
\end{equation}
The indicator term is the sign correction inherited from \(\beta_0\).
Indeed, because \(y=\gamma/\sqrt{1+\beta^2}\), the condition \(x_0y<0\)
is equivalent to \(\gamma/x_0<0\). Hence,
\begin{equation}
  2\beta_0=\mathbf{1}_{\{\gamma/x_0<0\}}.
\end{equation}
Finally, substituting \eqref{eq:lemma_J_infty} and
\eqref{eq:lemma_Jx0_final} into \eqref{eq:lemma_I_split} gives
\eqref{eq:lemma_I_general_closed}, completing the proof of Lemma~1.
\bibliographystyle{IEEEtran}
\bibliography{bibliography.bib}
\end{document}